\documentclass[journal]{IEEEtran}
\IEEEoverridecommandlockouts
\usepackage[export]{adjustbox}
\usepackage{setspace}
\usepackage{cite}
\usepackage{amsmath,amssymb,amsfonts,mathtools,bm}
\usepackage{amsthm}
\usepackage{algorithm}
\usepackage{graphicx}
\usepackage{textcomp}
\usepackage{xcolor,float}
\usepackage{tabularx}
\usepackage{textcomp}
\usepackage{diagbox}
\usepackage{svg}
\usepackage{booktabs}
\usepackage{subfigure}
\usepackage{hyperref}
\usepackage{algpseudocode}
\usepackage{graphicx}
\usepackage{makecell}
\usepackage{amsthm}
\usepackage{caption}
\usepackage{sidecap}
\usepackage{microtype}
\usepackage{booktabs} 
\usepackage{float}
\usepackage{adjustbox} 
\usepackage{amsmath}
\usepackage{amssymb}
\usepackage{colortbl}
\usepackage{makecell}
\usepackage{float}
\usepackage{url}
\usepackage{balance}
\usepackage{bbding}
\usepackage{hyperref}
\usepackage{graphicx}
\usepackage{sidecap}
\usepackage{longtable}
\usepackage{microtype}
\usepackage{graphicx}
\usepackage{booktabs} 
\usepackage{multirow}
\usepackage{array}
\usepackage{subcaption}
\usepackage{graphicx}
\usepackage{makecell}
\usepackage{amsmath}
\usepackage{amssymb}

\usepackage{tikz}
\usepackage[final]{changes}

\newtheorem{theorem}{Theorem}
\newtheorem{problem}{Problem}
\newtheorem{lemma}{Lemma}
\usepackage{ifthen}
\newboolean{showchanges}
\usepackage{booktabs}
\setboolean{showchanges}{True}  

\graphicspath{{./images/}}

\newcommand{\add}[1]{%
    \ifthenelse{\boolean{showchanges}}%
        {\textcolor{blue}{#1}}
        {#1\relax}
}

\definecolor{lime}{HTML}{A6CE39}
\DeclareRobustCommand{\orcidicon}{%
    \begin{tikzpicture}
    \draw[lime, fill=lime] (0,0) 
    circle [radius=0.16] 
    node[white] {{\fontfamily{qag}\selectfont \tiny ID}};    \draw[white, fill=white] (-0.0625,0.095) 
    circle [radius=0.007];    \end{tikzpicture}
    \hspace{-2mm}}
\foreach \x in {A, ..., Z}{%
    \expandafter\xdef\csname orcid\x\endcsname{\noexpand\href{https://orcid.org/\csname orcidauthor\x\endcsname}{\noexpand\orcidicon}}
    }

\allowdisplaybreaks
\renewcommand{\baselinestretch}{1}

\begin{document}

\title{Graph Neural Network-Based Multicast Routing for On-Demand Streaming Services in 6G Networks}
\author{Xiucheng Wang\orcidA{},~\IEEEmembership{Graduate Student Member,~IEEE,}
        Zien Wang\orcidB{},
        Nan Cheng\orcidC{},~\IEEEmembership{Senior Member,~IEEE,} \\
        Wenchao Xu\orcidD{},~\IEEEmembership{Member,~IEEE,}
        Wei Quan\orcidE{},~\IEEEmembership{Senior Member,~IEEE,}
        and Xuemin (Sherman) Shen\orcidF{},~\IEEEmembership{Fellow,~IEEE}
\thanks{
\par Xiucheng Wang, Zien Wang and Nan Cheng are with School of Telecommunications Engineering, Xidian University, Xi’an, 710071, China (e-mail: \{xcwang\_1, zewang\_1\}@stu.xidian.edu.cn; dr.nan.cheng@ieee.org). \textit{(Xiucheng Wang and Zien Wang contributed equally to this work.)(Corresponding author: Nan Cheng.)}
\par Wenchao Xu is with Division of Integrative Systems and Design, The Hong Kong University of Science and Technology, 999077, Hong Kong SAR (e-mail: wenchaoxu@ust.hk).
\par Quan Wei is with School of Electronic and Information Engineering, Beijing Jiaotong University, Beijing, 100044, China(e-mail: dr.wei.quan@ieee.org).
\par Xuemin (Sherman) Shen is with the Department of Electrical and Computer Engineering, University of Waterloo, Waterloo, N2L 3G1, Canada (e-mail: sshen@uwaterloo.ca).
}        
}

    \maketitle

\IEEEdisplaynontitleabstractindextext

\IEEEpeerreviewmaketitle

\begin{abstract}
The increase of bandwidth-intensive applications in sixth-generation (6G) wireless networks, such as real-time volumetric streaming, and multi-sensory extended reality, demands intelligent multicast routing solutions capable of delivering differentiated quality-of-service (QoS) at scale. Traditional shortest-path and multicast routing algorithms are either computationally prohibitive or structurally rigid, and they often fail to support heterogeneous user demands, leading to suboptimal resource utilization. Neural network-based approaches, while offering improved inference speed, typically lack topological generalization and scalability. To address these limitations, this paper presents a graph neural network (GNN)-based multicast routing framework that jointly minimizes total transmission cost and supports user-specific video quality requirements. The routing problem is formulated as a constrained minimum-flow optimization task, and a reinforcement learning algorithm is developed to sequentially construct efficient multicast trees by reusing paths and adapting to network dynamics. A graph attention network (GAT) is employed as the encoder to extract context-aware node embeddings, while a long short-term memory (LSTM) module models the sequential dependencies in routing decisions. Extensive simulations demonstrate that the proposed method closely approximates optimal dynamic programming-based solutions while significantly reducing computational complexity. The results also confirm strong generalization to large-scale and dynamic network topologies, highlighting the method’s potential for real-time deployment in 6G multimedia delivery scenarios. \added{Code is available at \url{https://github.com/UNIC-Lab/GNN-Routing}.}
\end{abstract}


\section{Introduction}
The advent of sixth-generation (6G) wireless networks is catalyzing a profound transformation in multimedia services, enabling a broad spectrum of bandwidth-intensive and latency-sensitive applications. Emerging use cases such as holographic telepresence, real-time volumetric video streaming, and multi-sensory extended reality (XR) are redefining user expectations and imposing unprecedented demands on network infrastructures. For instance, holographic telepresence services require traffic densities on the order of 1–10 Tbps/km² \cite{Menglan}, while real-time volumetric video applications may necessitate peak data rates exceeding 100 Gbps per user \cite{DBLP:journals/corr/abs-1906-00741}. Similarly, XR applications involve synchronized audiovisual and haptic feedback, further intensifying the quality-of-service (QoS) requirements in terms of reliability, latency, and throughput \cite{Shen}\cite{YHE}, particularly in mobile and densely populated environments. This paradigm shift is further illustrated by the explosive growth in global multimedia traffic. However, such growth also brings substantial technical challenges. First, the network must support heterogeneous QoS demands, as multiple users in the same session may request vastly different resolutions—from 360p for low-end devices to 8K for immersive viewing experiences \cite{11017513}. Second, to ensure economic viability and environmental sustainability, 6G infrastructures must minimize redundant transmissions and optimize flow efficiency to reduce energy consumption and system overheads, particularly in multicast and multi-hop scenarios \cite{radiodiff}, \cite{Chukhno}. The dual imperatives of service differentiation and flow minimization introduce significant complexity to video transmission routing, especially under the dynamic and large-scale conditions envisioned for 6G. This necessitates the development of intelligent, flexible, and scalable routing mechanisms that can simultaneously meet personalized user demands and minimize the global transmission cost.

Despite decades of advancements in network routing, conventional algorithms struggle to meet the multifaceted demands of video transmission in 6G environments. Classic shortest-path routing algorithms, such as Dijkstra’s or Bellman-Ford, are effective in identifying low-cost routes between a single source and a single destination in polynomial time. However, when extended to scenarios involving multiple concurrent users requesting the same video stream, these algorithms reveal significant inefficiencies. Specifically, the independent computation of shortest paths for each user fails to exploit path reuse opportunities, resulting in redundant transmissions that inflate the overall network flow and energy consumption \cite{Chen2004Multipath}. This inefficiency becomes particularly critical as live streaming and interactive video applications continue to proliferate, with the global live streaming market alone expected to surpass USD 3.7 trillion by 2030 \cite{latreche2025applicationsenvisagednewgeneration}. To address the path reuse problem, traditional routing solutions often adopt multicast routing algorithms that construct a shared distribution tree from the source to all destination nodes. Techniques based on dynamic programming or Steiner tree formulations aim to minimize the aggregate transmission cost by finding a weighted minimum-cost tree. While theoretically optimal, these solutions are NP-hard\cite{Messmer2010Real-Time}, and their computational complexity renders them impractical for real-time inference in large-scale and highly dynamic 6G networks\cite{Chen2012On}. Moreover, their reliance on global topology knowledge and iterative computation makes them unsuitable for distributed and latency-constrained edge environments\cite{Chi2018Live}\cite{YHE2}. More critically, these traditional multicast algorithms generally assume homogeneous service demands, treating all users as requesting the same quality of service. In real-world scenarios, however, users often demand heterogeneous video resolutions—for example, a VR headset may require ultra-high-resolution data while a mobile phone only needs a low-resolution version of the same content. Conventional algorithms typically construct a single delivery path based on the highest QoS requirement, which leads to resource over-provisioning for users with lower requirements. This not only wastes bandwidth but also diminishes overall system efficiency, especially in resource-constrained settings\cite{Li2020Jointly}\cite{Zou2011Prioritized}. As 6G networks evolve toward service-aware and resource-efficient architectures, the limitations of traditional routing algorithms become more pronounced. There is a clear need for novel routing paradigms that are capable of accommodating user diversity, topological dynamics, and computational tractability—a combination that traditional approaches are fundamentally ill-equipped to provide.
\begin{figure*}[!t]
    \centering
    \includegraphics[width=0.92\linewidth]{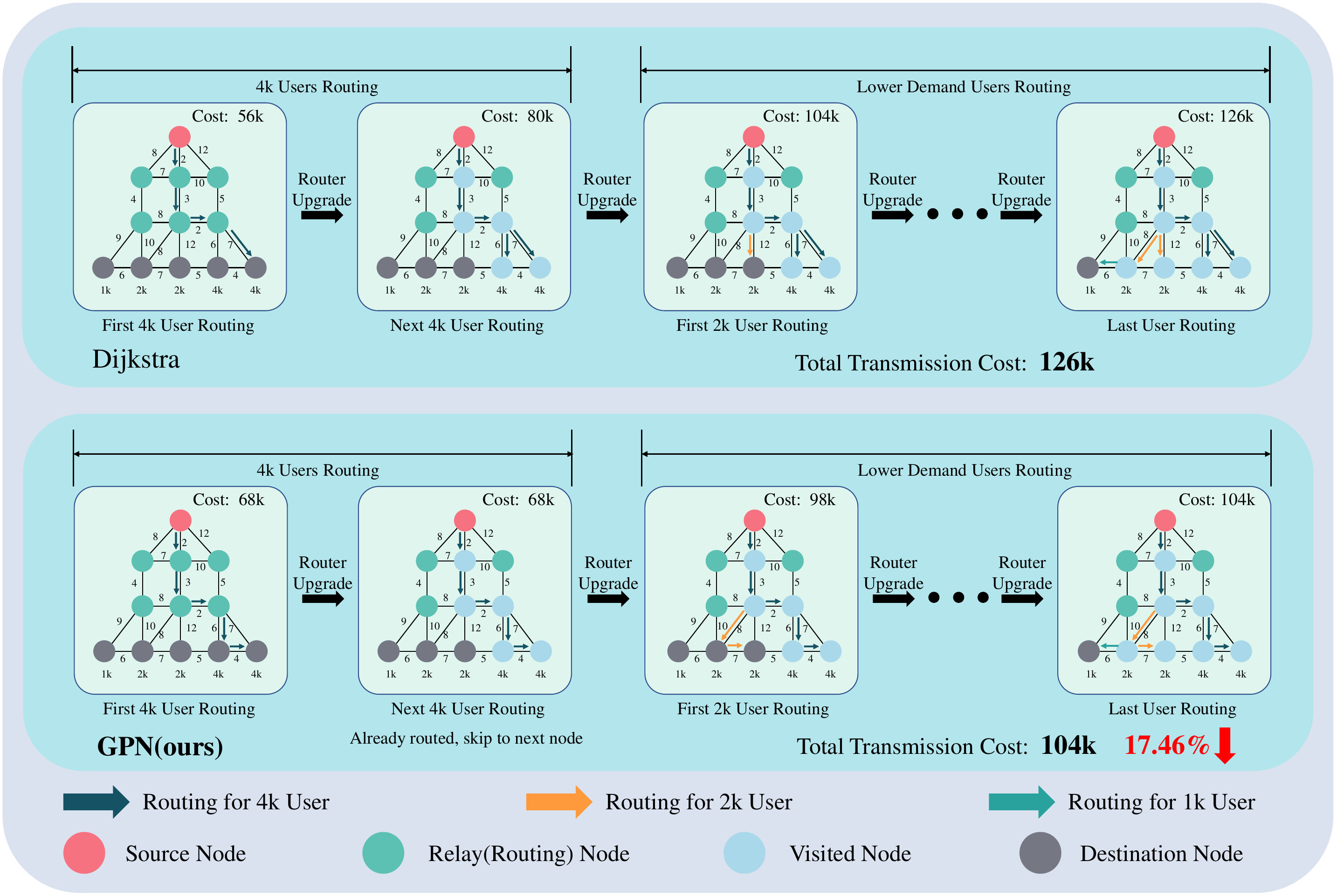}
    \caption{Routing path comparison example under sequential user arrivals. Each row illustrates the step-by-step routing process of a different algorithm (Dijkstra vs. GPN(ours)), with the total transmission cost accumulated across all user demands.}
    \label{fig:path_compare_smallusers}
\end{figure*}
To overcome the scalability and inference latency issues associated with traditional routing algorithms, researchers have turned to neural network (NN)-based approaches for data-driven routing decision-making. These models aim to approximate optimal routing policies through learning, enabling real-time inference with fixed computational complexity—an appealing feature for highly dynamic 6G environments\cite{10387423}. For instance, multi-layer perceptron (MLP)-based methods have been proposed to map network states directly to routing decisions, achieving constant-time inference once trained. These methods can significantly accelerate the routing process in static or moderately dynamic networks by avoiding combinatorial searches during runtime\cite{Du2020Deep}. Building on this, other studies have incorporated convolutional neural networks (CNNs) to capture the spatial and structural patterns of network topologies. By representing the connection matrix of the network as an image-like structure, CNNs can extract hierarchical features that correlate with efficient routing paths. These models show improved performance over basic MLPs by more effectively learning from the topological layout of the network. Such approaches have demonstrated the feasibility of learning-based routing, especially in simulated or small-scale networks with stable conditions. However, despite their promise, traditional NN-based routing models exhibit critical limitations when deployed in realistic 6G network scenarios. One major bottleneck is their lack of structural scalability. Standard neural architectures, such as MLPs and CNNs, require fixed input and output dimensions, which constrains them to networks with a predetermined number of nodes and links\cite{Langpoklakpam2023Review}. Any change in the number of users or modification in the network topology necessitates retraining the model or re-engineering the architecture, significantly undermining their applicability in dynamic 6G networks where users frequently join, leave, or move. Moreover, these models struggle with generalization across unseen topologies. Since their learning is often based on absolute positional or indexing information, their inference capabilities degrade sharply when presented with a network graph that differs from the training data\cite{Tang2024Is}. This results in poor robustness, especially when the routing model is expected to operate across multiple regions or deployment scenarios. Another key limitation is that many existing NN-based methods are not inherently aware of graph-structured data. Unlike human-designed algorithms that leverage the relational nature of networks, such as trees, paths, and flows, standard neural networks treat the input as unstructured data, failing to fully exploit the relational inductive bias essential for networked systems\cite{Orucu2024Towards}. Consequently, their learning and inference efficiency is suboptimal compared to architectures explicitly designed for structured domains. In light of these issues, there is a growing consensus that routing algorithms for 6G networks must not only be data-driven but also graph-aware, topology-adaptive, and scalable. These requirements point to the need for more advanced learning architectures capable of naturally handling dynamic, large-scale, and heterogeneous network environments.

Motivated by the shortcomings of both traditional and conventional neural network-based routing methods, graph neural networks (GNNs) have emerged as a compelling alternative for routing optimization in 6G multimedia networks. GNNs are inherently designed to operate on graph-structured data, making them well-suited to model the topological dependencies, node interactions, and spatial correlations present in communication networks\cite{Jiang2024Graph}. Unlike MLPs or CNNs, GNNs can flexibly accommodate variations in network size, connectivity, and user count without requiring architectural changes or retraining, thereby offering unprecedented scalability and adaptability\cite{Galmés2021Scaling}. One of the most salient features of GNNs is their message passing mechanism, which allows each node to iteratively exchange and aggregate information from its neighbors. This local information propagation not only mirrors the structure of many distributed routing algorithms but also enables GNNs to learn efficient representations of both global network context and local path preferences\cite{Lent2022Dynamic}. Importantly, message passing is structure-agnostic—the number of users, UAVs, or relay nodes can change dynamically, and the same GNN model can generalize to new topologies\cite{Liu2021Routing}. This contrasts sharply with fixed-structure NNs, which are limited to static, pre-defined graphs. In the context of 6G video transmission, where users may request heterogeneous QoS levels and network topologies evolve in real time, the flexibility and efficiency of GNNs become even more valuable. For instance, GNNs can naturally support differentiated routing paths that reflect user-specific video resolution requirements, while still minimizing overall network cost by leveraging shared routes where appropriate\cite{Changala2024Enhancing}. Additionally, the number of trainable parameters in GNN models is typically far lower than those in fully connected or convolutional networks, enabling lightweight deployment on resource-constrained platforms such as UAVs or edge nodes. To address the unique challenges of on-demand, scalable, and QoS-aware video streaming in 6G networks, this paper proposes a novel GNN-based routing framework. Specifically, we design a routing algorithm that jointly optimizes video stream dissemination paths and network resource efficiency, while supporting user-specific service differentiation. The proposed method maintains linear computational complexity with respect to the number of users and network nodes, making it suitable for deployment in large-scale and rapidly evolving network environments\cite{Guliyev2024D3-GNN}. By leveraging the representational power and structural generalization of GNNs, the proposed approach bridges the gap between routing optimality, inference efficiency, and scalability, providing a foundational step toward intelligent, adaptive, and service-aware video transmission in next-generation wireless systems. The main contribution of this paper are summarized as follows.
\begin{enumerate}
    \item To the best of our knowledge, this work is the first to investigate the problem of multicast routing for live video streaming in 6G networks with differentiated user demands. Specifically, we address how to construct a routing path that connects a single source node to multiple user nodes—each requesting different service qualities—while minimizing the total transmission cost. We formulate this as a minimum-flow optimization problem with inflow constraints, capturing both path reuse and user-specific QoS requirements.
    \item A GNN-based routing framework is proposed to efficiently solve the formulated problem. By integrating graph attention mechanisms and modeling flow reuse, the algorithm achieves linear time complexity with respect to the number of nodes as $\mathcal{O}(n)$. Leveraging the message-passing paradigm of GNNs, the trained model generalizes across arbitrary network topologies without retraining, offering strong scalability and robustness to topological variations.
    \item Extensive simulations are conducted to validate the effectiveness of the proposed method. Results demonstrate that the GNN-based approach achieves performance comparable to exhaustive search methods, while significantly reducing computational overhead. Moreover, the method exhibits excellent scalability to large-scale networks and strong adaptability to dynamic environmental conditions, highlighting its suitability for real-time deployment in 6G multimedia routing scenarios.
\end{enumerate}

\section{Related Works and Preliminary}
\subsection{Graph Neural Networks}
GNNs have emerged as a powerful class of neural models designed to process data represented in graph structures, which are ubiquitous in various domains such as social networks, knowledge graphs, molecular chemistry, and increasingly, wireless communication networks\cite{Das2025Opportunistic}. A graph is formally defined as $\mathcal{G} = (\mathcal{V}, \mathcal{E})$, where $\mathcal{V}$ is the set of $N$ nodes and $\mathcal{E} \subseteq \mathcal{V} \times \mathcal{V}$ is the set of edges. Each node $v_i \in \mathcal{V}$ may be associated with a feature vector $\bm{x}_i \in \mathbb{R}^d$, and the graph topology is typically encoded in an adjacency matrix $\bm{A} \in \mathbb{R}^{N \times N}$. Early developments in GNNs were rooted in spectral graph theory, where convolution operations are defined in the frequency domain by leveraging the graph Laplacian\cite{Stachenfeld2020Graph}\cite{Yang2021Improving}. Given a graph with a symmetric adjacency matrix $\bm{A}$ and a degree matrix $\bm{D}$ (where $D_{ii} = \sum_j A_{ij}$), the normalized graph Laplacian is defined as $\bm{L} = \bm{I} - \bm{D}^{-1/2} \bm{A} \bm{D}^{-1/2}$. Spectral graph convolution can then be expressed via the eigendecomposition of $\bm{L} = \bm{U} \bm{\Lambda} \bm{U}^\top$, where $\bm{U}$ contains the eigenvectors and $\bm{\Lambda}$ is a diagonal matrix of eigenvalues. For a graph signal $\bm{x}$, the convolution with a filter $g_\theta$ is defined as follows.
\begin{align}
    g_\theta * \bm{x} = \bm{U} g_\theta(\bm{\Lambda}) \bm{U}^\top \bm{x},
\end{align}
where $g_\theta(\cdot)$ is a spectral filter function parameterized by learnable weights $\theta$. While elegant, spectral methods require the computation of the full eigendecomposition of the Laplacian, which is computationally prohibitive for large graphs and lacks generalizability to dynamically changing topologies\cite{geisler2024spatiospectralgraphneuralnetworks}.

To address these limitations, message passing graph neural networks (MP-GNNs) have become the dominant paradigm. MP-GNNs operate directly in the spatial domain and are inherently localized, scalable, and topology-independent\cite{Olshevskyi2024Fully}\cite{veličković2018graphattentionnetworks}. The key idea is to iteratively update each node’s feature vector by aggregating information from its neighbors through a two-step process: message computation and node update. At the $k$-th layer, the generic update rule for node $v_i$ is given as follows.
\begin{align}
    \bm{h}_i^{(k)} = \phi^{(k)}\left( \bm{h}_i^{(k-1)}, \varphi^{(k)}\left( \left\{ \bm{h}_j^{(k-1)} \mid j \in \mathcal{N}(i) \right\} \right) \right),
\end{align}
where $\bm{h}_i^{(k)}$ is the hidden representation of node $i$ at layer $k$, $\mathcal{N}(i)$ denotes the set of neighbors of node $i$, and $\varphi^{(k)}(\cdot)$ and $\phi^{(k)}(\cdot)$ are differentiable functions which typically are neural networks. A common instantiation is the Graph Convolutional Network (GCN), where the propagation rule simplifies as follows.
\begin{align}
    \bm{H}^{(k)} = \sigma \left( \tilde{\bm{D}}^{-1/2} \tilde{\bm{A}} \tilde{\bm{D}}^{-1/2} \bm{H}^{(k-1)} \bm{W}^{(k)} \right),
\end{align}
where $\tilde{\bm{A}} = \bm{A} + \bm{I}$ as the adjacency matrix with added self-loops, $\tilde{\bm{D}}$ as its degree matrix, $\bm{W}^{(k)}$ the trainable weight matrix, and $\sigma(\cdot)$ an activation function such as ReLU. This structure allows information to propagate across multiple hops as layers are stacked, thereby enabling long-range dependency modeling. Compared with traditional neural network architectures, GNNs naturally incorporate relational inductive bias through the graph structure, making them particularly effective in domains where interactions among entities are as important as the entities themselves. In wireless communication and routing problems, nodes can represent users, UAVs, or base stations, and edges capture connectivity or flow constraints. Thus, GNNs offer a powerful tool for designing adaptive, distributed, and structure-aware algorithms capable of generalizing across variable network topologies.

\subsection{Multi-casting Routing Methods}
Multicast routing has long been recognized as a critical mechanism in communication networks, enabling efficient point-to-multipoint data dissemination. In traditional wired infrastructure, early approaches were built upon the theory of multicast trees, particularly the Steiner Tree formulation, which minimizes the total transmission cost from a single source to multiple destinations. However, due to its NP-hard complexity, most practical protocols opted for heuristics and approximations\cite{Tshakwanda2024Advancing}. Protocols such as Distance Vector Multicast Routing Protocol (DVMRP), Multicast OSPF (MOSPF), and Protocol Independent Multicast (PIM) emerged as foundational methods in IP multicast routing. DVMRP and PIM Dense Mode relied on a flood-and-prune mechanism, suitable for densely populated receiver scenarios, while MOSPF and PIM Sparse Mode introduced join-based and shared-tree strategies to accommodate sparse group distributions and improve scalability. To support efficient delivery across provider networks, multicast was also adapted to MPLS-based backbones through mechanisms like point-to-multipoint Label Switched Paths (P2MP LSPs) using RSVP-TE or multipoint LDP. More recently, Bit Index Explicit Replication (BIER) has been proposed to eliminate per-group state in the core, offering stateless multicast through bit-vector encoded packet headers. In wireless and mobile ad hoc networks (MANETs), the dynamic topology and lack of fixed infrastructure require new multicast designs. Tree-based protocols, such as MAODV, extended unicast routing methods to multicast delivery, while mesh-based schemes like ODMRP provided greater resilience to mobility-induced link failures. Other protocols, such as AMRIS and CAMP, combined tree and mesh features to improve robustness and efficiency. Vehicular networks (VANETs) and wireless sensor networks (WSNs) introduced further specialization. VANET multicast often adopted geocast techniques using geographic positions and road topology to manage high mobility\cite{Nahar2023SDCast}, whereas WSN multicast prioritized energy efficiency and often leveraged group Steiner tree heuristics under duty-cycle constraints.

With the evolution of programmable networks, software-defined networking (SDN) has provided a centralized architecture to manage multicast more flexibly. SDN controllers can compute optimal or near-optimal multicast trees in real time and adjust them dynamically based on traffic and network conditions. This global view enables traffic engineering, congestion avoidance, and quality-of-service enforcement at scale\cite{Ye2023DHRL-FNMR}. Furthermore, SDN supports innovative approaches such as multi-group shared trees and label-based forwarding, significantly reducing multicast state and improving adaptability\cite{Zhao2022DRL-M4MR}. Parallel to SDN advancements, machine learning and reinforcement learning have emerged as powerful tools for multicast optimization. Deep reinforcement learning agents have been used to construct multicast trees that adapt to changing topologies and traffic demands, often outperforming heuristic methods in throughput and delay\cite{Lu2023Hawkeye}. Genetic algorithms, ant colony optimization, and other metaheuristic techniques have also been applied to address multicast problems with multiple conflicting objectives, such as cost, delay, reliability, and energy consumption\cite{Rathan2021Q-Learning}\cite{Zheng2012An}. These AI-driven approaches enable more intelligent and context-aware routing strategies, particularly in complex or large-scale networks where traditional protocols may struggle. Collectively, the body of research on multicast routing demonstrates a progressive shift from protocol-driven to learning-driven paradigms, spanning both infrastructure-based and ad hoc networks. However, challenges remain in achieving low-latency, QoS-aware, and scalable multicast delivery in the face of increasing network dynamics, service differentiation, and resource constraints, highlighting the need for new models that are both structure-aware and computationally efficient.

\section{System Model and Problem Formulation}\label{sec-3}
In this paper, we consider a generalized minimum-cost flow problem in which a single source node disseminates data to multiple destination nodes over a network, while satisfying heterogeneous flow demands across receivers. The underlying network is modeled as an undirected weighted graph $\mathcal{G} = (\mathcal{V}, \mathcal{E})$, where $\mathcal{V}$ denotes the set of $M$ nodes and $\mathcal{E} \subseteq \mathcal{V} \times \mathcal{V}$ represents the set of bidirectional communication links. Each edge $(i, j) \in \mathcal{E}$ is associated with a non-negative weight $e_{(i,j)}$, representing the unit transmission cost of flow between nodes $i$ and $j$. The node set $\mathcal{V}$ is partitioned into three disjoint subsets: the source node set $\mathcal{V}_{s}$, the set of destination nodes $\mathcal{V}_{d}$, and the set of relay (routing) nodes $\mathcal{V}_{r}$, such that $\mathcal{V} = \mathcal{V}_s \cup \mathcal{V}_d \cup \mathcal{V}_r$. The cardinalities are $|\mathcal{V}_s| = 1$, $|\mathcal{V}_d| = K$, and $|\mathcal{V}| = M$. Let the single source node be denoted by $s \in \mathcal{V}_s$, and the destination nodes be indexed by $k = 1, 2, \dots, K$. We denote the flow variable from node $i$ to node $j$ as $f_{(i,j)} \geq 0$, which represents the amount of flow allocated along edge $(i,j)$. The demand vector is given by $\bm{x} = [x^1, x^2, \dots, x^K]^\top$, where $x^k$ specifies the minimum required inflow at destination node $k \in \mathcal{V}_d$. The decision variables in this problem are the set of flows $f_{(i,j)\in\mathcal{E}}$ that define how the total transmission is distributed across the network links, subject to flow conservation and demand satisfaction constraints.

The problem is formulated as follows.
\begin{problem}\label{p1}
    \begin{align}
        &\min_{\bm{f}} \sum_{i\in \mathcal{V}}\sum_{j\in \mathcal{V}} e_{(i,j)}f_{(i,j)},\label{obj}\\
        &s.t.\;  \max_{j\in\mathcal{V}} f_{(i,j)} \leq \max_{k\in\mathcal{V}} f_{(i,j)}, &&\forall i\in\mathcal{V}/\mathcal{V}_{s},\tag{\ref{obj}a}\label{c1}\\
        &\quad\;\; \max_{j\in \mathcal{V}} f_{(i,j)}\geq \max_{k} x^{k},&&\forall k\in\mathcal{V}_{d} \wedge i \in\mathcal{V}_{s},\tag{\ref{obj}b}\label{c2}\\
        &\quad\;\;x^{j}\leq \max_{i\in \mathcal{V}\setminus\mathcal{V}_{d}} f_{(i,j)}, &&\forall j\in\mathcal{V}_{d} ,\tag{\ref{obj}c}\label{c3}\\
        &\quad\;\; f_{(i,j)} \geq 0, &&\forall i,j \in \mathcal{V},\tag{\ref{obj}d}\label{c4}\\
        &\quad\;\; f_{(i,j)} = 0, &&\forall e_{(i,j)}\notin \mathcal{E},\tag{\ref{obj}e}\label{c5}
    \end{align}
\end{problem}
The objective in \eqref{obj} is to minimize the total flow cost across all edges of the network. Constraint \eqref{c1} ensures flow conservation at all intermediate and destination nodes, i.e., the outgoing flow must not exceed the total incoming flow for any node except the source. Constraint \eqref{c2} enforces that the total outflow from the source is at least sufficient to meet the maximum individual demand among destination nodes. Constraint \eqref{c3} guarantees that each destination node receives a flow no less than its required demand. Constraints \eqref{c4} and \eqref{c5} enforce non-negativity of flows and topological feasibility, ensuring that flow is only allowed on existing edges.

\added{In our target deployment, each routing node abstracts a city- or region-level controller in the provider core/metro network, where link weights are administratively configured and remain stable over minutes to hours (e.g., OSPF/IS-IS TE metrics\cite{sllame}, MPLS/SRv6 path costs\cite{liangbocheng}). Hence, the unit transmission cost $e_{(i,j)}$ primarily captures geography-driven transport cost, peering/egress tariff, and energy per bit; these factors vary slowly and can be treated as fixed during a routing session.} 

\added{
We model the wireless access segment as \emph{piecewise-constant} over short control windows: although radio conditions fluctuate, configurations are computed from the most recent measurements per window and applied until the next update; thus we optimize on the latest snapshot and trigger event-driven re-planning upon significant deviations. This time-scale separation preserves per-snapshot tree optimality, while inter-window dynamics are absorbed by subsequent re-optimizations; snapshot-based TE under variability follows the same regime~\cite{figret_sigcomm24}. Our sub-second inference on practical graph sizes also ensures that such re-plans can be executed promptly in deployment.
}

\section{GNN Based Multi-cast Method}
\subsection{Problem Analysis}
To gain deeper insight into the mathematical structure of the optimization problem defined in Problem \ref{p1}, we investigate the topological and flow properties of its optimal solution. Understanding these structural characteristics not only facilitates efficient algorithm design but also motivates the application of graph-based learning architectures in later sections. At its core, the optimization problem seeks to deliver flow from a single source to multiple destination nodes, each with heterogeneous demand, while minimizing the total transmission cost over a given network topology. Unlike unicast or uniform multicast problems, the presence of differentiated QoS requirements across destination nodes introduces additional complexity. Nevertheless, it can be shown that the optimal solution exhibits a highly regular topological structure that simplifies the underlying problem. Specifically, we establish that the set of links carrying non-zero flow in the optimal solution forms a tree rooted at the source node, spanning all destination nodes as leaves. This structural insight is formalized in the following theorem.
\begin{theorem}\label{theorem-1}
For links carrying flow, i.e., where $f_{(i,j)} > 0$ and $e_{(i,j)} \in \mathcal{E}$, these links form a tree structure with the source node as the root and all destination nodes as the leaf nodes. The flow on each link follows the direction from the root node to the leaf nodes, corresponding to increasing depth in the tree. The minimum flow that satisfies the outflow requirements of all destination nodes must conform to this tree structure.
\end{theorem}
\begin{proof}
    The tree structure described in Theorem \ref{theorem-1} trivially satisfies constraints \eqref{c2}, \eqref{c4}, and \eqref{c5}. To satisfy constraints \eqref{c1} and \eqref{c3}, there must exist a connected path from the source node to each destination node. The subgraph formed by the flow-carrying links must therefore be connected. If a disconnected subgraph exists, any component not connected to the source node cannot deliver flow to all destination nodes. Removing such isolated components would reduce the total cost while maintaining constraint feasibility, contradicting optimality. Furthermore, suppose the flow-carrying subgraph contains a cycle. Let $v$ be the node in the cycle that is closest to the source. Eliminating the incoming link to $v$ along the cycle preserves connectivity and all constraint conditions, while reducing redundant flow, yielding a lower-cost solution. Similarly, if two nodes $u$ and $v$ are connected by multiple disjoint paths, any flow along one path can be removed without violating constraints, again improving the objective value. Lastly, if the tree includes a non-destination leaf node, the branch leading to that node can be pruned without affecting destination demands, further reducing cost. Thus, the flow-carrying subgraph must be a minimal connected acyclic structure spanning all destination nodes from the source—i.e., a tree rooted at the source.
\end{proof}

The tree structure not only ensures minimum redundancy but also reveals how flow is shared and reused across paths toward multiple destination nodes. To further characterize the optimal flow allocation along shared links, we state the following lemma.

\begin{lemma}\label{lemma-1}
In the optimal solution of Problem \ref{p1}, if a link is shared to transmit flow from the source node to multiple destination nodes, the flow size on this link equals the maximum outflow demand among those destination nodes.
\end{lemma}

\begin{proof}
If the flow along a shared link is less than the highest downstream demand, then according to Theorem \ref{theorem-1} and constraint \eqref{c1}, the destination node with the largest demand cannot receive a sufficient amount of flow, thereby violating constraint \eqref{c3}. Therefore, the link must carry at least the maximum demand among all downstream destination nodes, and any excess flow would violate the optimality condition.
\end{proof}

These structural observations play a critical role in guiding the design of learning-based solvers. In particular, the tree-like topology and flow propagation characteristics suggest that a graph neural network can effectively infer link usage and flow allocation by leveraging the hierarchical dependencies within the network. As demonstrated in the following sections, incorporating this domain knowledge into the GNN model improves both learning efficiency and generalization across variable network topologies.

\subsection{Policy-Gradient Multicast Routing Method}
\added{Building on the structural results in Section~\ref{sec-3}, especially Theorem~\ref{theorem-1} and Lemma~\ref{lemma-1}, we design a reinforcement-learning multicast routing algorithm that exploits the \emph{flow-reuse} nature of the optimal solution. Let the network be $\mathcal{G}=(\mathcal{V},\mathcal{E})$ with link cost $e(u,v)$ and user demands $\{x_k\}$. Users are processed in descending demand order, i.e., $(u_{(1)},\ldots,u_{(K)})$ with $x_{(1)}\ge\cdots\ge x_{(K)}$. For $k=1$, we build a minimum-cost path from $u_{(1)}$ to the source. For each $k>1$, we grow a path from $u_{(k)}$ until it reaches any previously visited node; the newly constructed path is then merged and downstream flow is reused. This sequential scheme is consistent with the tree-structured optimality and the shared-edge maximum downstream demand property, and it avoids redundant transmissions by construction. To account for global cost, we cast the per-user path construction as an episodic MDP rather than solving independent shortest paths. Naively selecting the locally shortest route for each user can raise the cumulative cost of later users by eliminating profitable reuse opportunities. In contrast, our policy-gradient formulation optimizes long-term return, explicitly trading off immediate transmission cost and downstream reuse.}

\subsubsection{State.}
\added{At user $k$ and decision step $t$, the state is}
\begin{align}
    s_t=\big(\mathcal{G},\,\mathcal{V}^{(k)}_{\text{inflow}},\,\mathcal{P}_t\big),
\end{align}
\added{where $\mathcal{V}^{(k)}_{\text{inflow}}\subseteq\mathcal{V}$ collects visited/inflow-enabled nodes created by users $1{:}k{-}1$, and $\mathcal{P}_t\subseteq\mathcal{V}$ is the current partial path for $u_{(k)}$ whose last node is $u_t$. The candidate set is}
\begin{align}
    \mathcal{N}_t \triangleq \{\, v\in\mathcal{V} : (u_t,v)\in\mathcal{E}\, \}.
\end{align}
\subsubsection{Action.}
\added{Choose the next node $v_t \in \mathcal{N}_t \cup \mathcal{V}^{(k)}_{\text{inflow}}$.
If $v_t\in\mathcal{V}^{(k)}_{\text{inflow}}$, the episode terminates and the completed path $\mathcal{P}_k$ is merged:}
\begin{align}
    &\mathcal{V}^{(k+1)}_{\text{inflow}} \;=\; \mathcal{V}^{(k)}_{\text{inflow}} \,\cup\, \mathcal{P}_k, \\ &\mathcal{P}_{t{+}1}=\emptyset.
\end{align}
\added{Otherwise, set $\mathcal{P}_{t{+}1}=\mathcal{P}_t\cup\{v_t\}$ and continue.}

\subsubsection{Reward and Objective.}
\added{At each step}
\begin{align}
    r_t \;=\; -\, x_{(k)} \, e(u_t,v_t),
\end{align}
\added{so the return is $R=\sum_t \gamma^t r_t$ with discount $\gamma\in(0,1)$. Maximizing $\mathbb{E}[R]$ aligns the policy with minimizing expected total transmission cost while preserving future reuse potential.}

\subsection{Graph Attention Based RL Method}
To learn efficient routing decisions under the sequential, demand-aware multicast formulation, we propose a graph policy network (GPN) based on a graph attention network (GAT) encoder and an LSTM-based path history aggregator\cite{Zheng2012An}\cite{10.1162/neco.1997.9.8.1735}. Since we apply reinforcement learning to guide node-wise routing decisions, the model cannot receive immediate feedback after each action. The final reward—total transmission cost—is only available after a complete path is formed. Moreover, due to the discrete selection and max-aggregation over path lengths, the objective function is non-differentiable with respect to the neural network parameters. Therefore, we adopt a policy gradient method to update the GPN parameters. The GPN consists of two primary components: a GAT-based encoder and an LSTM-based path aggregator. The encoder is responsible for learning node embeddings that reflect both local topology and demand context, while the decoder uses these embeddings to guide next-hop selection during path construction. 
\subsubsection{GAT Encoder}
The encoder is built upon the Graph Attention Network (GAT), which enhances traditional GNNs by learning attention weights for each neighbor rather than assigning them equal importance. For each node $i$ and its neighbor $j \in \mathcal{N}(i)$, the GAT computes attention coefficients $\alpha_{ij}$ as follows.
\begin{align}
e_{ij} &= \text{LeakyReLU}\left(\bm{a}^\top [\bm{W} \bm{x}_i , | , \bm{W} \bm{x}_j] \right), \\
\alpha_{ij} &= \frac{\exp(e_{ij})}{\sum_{k \in \mathcal{N}(i)} \exp(e_{ik})},
\end{align}
where $\bm{W}$ is a shared linear transformation, $\bm{a}$ is a trainable attention vector, $|$ denotes concatenation, and $\bm{x}_i$ is the input feature of node $i$. The updated node embedding $\bm{h}i$ is then computed as:
\begin{align}
\bm{h}i = \sigma\left( \sum{j \in \mathcal{N}(i)} \alpha_{ij} \bm{W} \bm{x}_j \right),
\end{align}
where $\sigma(\cdot)$ is an activation function, such as ReLU. This mechanism enables the encoder to attend differently to neighbors based on their relevance to the current graph state and demand structure.

\subsubsection{LSTM-Based Path History Aggregator}
To capture the sequential dependency of selected routing nodes, an LSTM is used to aggregate features from previously chosen nodes. This component allows the policy to consider the history of decisions made, which is essential since routing is constructed incrementally. At each routing step $t$, the LSTM state is updated as follows.
\begin{align}
\bm{h}^t, \bm{c}^t = \text{LSTM}(\bm{x}^t; \bm{h}^{t-1}, \bm{c}^{t-1}),
\end{align}
where $\bm{x}^t$ is the GAT-encoded feature of the $t$-th selected node, and $(\bm{h}^t, \bm{c}^t)$ denote the hidden and cell states of the LSTM, encoding the cumulative representation of the selected path. The initial input to the LSTM is the GAT embedding of the user node being routed.

\subsubsection{Attention-Based Decoder}
The decoder computes the policy distribution over the candidate next-hop nodes based on attention between the LSTM’s current hidden state and the GAT-encoded candidate node features. For each candidate node $v \in \mathcal{V}^{\text{next}}$, an attention score is computed as follows.
\begin{align}
p_v^t =
\begin{cases}
(\bm{h}^{t-1})^\top \tanh(W_2 \bm{x}_v + W_3 \bm{h}^{t-1}) & \text{if } v \notin \mathcal{P}_t, \\
-\infty & \text{otherwise},
\end{cases}
\end{align}
where $W_2$ and $W_3$ are trainable weight matrices, and $\mathcal{P}t$ denotes the set of nodes already in the selected path (to prevent loops). The softmax function is then applied to normalize the scores and obtain the probability distribution over actions as follows.
\begin{align}
\pi\theta(s_t, a_t) = \text{Softmax}(p_v^t).
\end{align}
The next-hop node is sampled from this distribution during training or selected greedily during inference. This selection strategy enables the policy to balance local cost (edge weight) and global path optimality learned through accumulated reward.

\begin{theorem}\label{theorem-2}
\added{On sparse graphs with $|\mathcal{E}|=\Theta(|\mathcal{V}|)$ and fixed $H,K$, the GA-based routing algorithm is at least
$\Omega(GPU\log|\mathcal{V}|)$ times more expensive than the proposed GPN method. If $G,P$ are treated as constants, the gap is $\Omega(|\mathcal{V}_d|\log|\mathcal{V}|)$.}
\end{theorem}

\begin{proof}
\added{GPN computes one graph encoding and reuses embeddings; thus
$T_{\text{GPN}}=\mathcal{O}(|\mathcal{V}|H^{2}K+|\mathcal{E}|HK)=\Theta(|\mathcal{V}|)$
when $H,K$ are fixed and $|\mathcal{E}|=\Theta(|\mathcal{V}|)$.
In our GA implementation, each generation evaluates $P$ individuals, and each fitness evaluation performs up to $|\mathcal{V}_d|$ weighted shortest-path runs.
Using Dijkstra, one run costs $\Theta((|\mathcal{E}|+|\mathcal{V}|)\log|\mathcal{V}|)=\Theta(|\mathcal{V}|\log|\mathcal{V}|)$ on sparse graphs.
Hence over $G$ generations:
$T_{\text{GA}}=\Theta(GP|\mathcal{V}_d|\,|\mathcal{V}|\log|\mathcal{V}|)$, so
$\frac{T_{\text{GA}}}{T_{\text{GPN}}}=\Omega(GP|\mathcal{V}_d|\log|\mathcal{V}|)$.
If $G,P$ are constants, this reduces to $\Omega(|\mathcal{V}_d|\log|\mathcal{V}|)$, where $G$ is the number of generations, $P$ the population size, $H$ the hidden dimension, and $K$ the number of attention heads.} 
\end{proof}

\added{
Our proposed GPN contains fewer than 3M parameters under our settings, resulting in a memory footprint of less than 50 MB during inference. This compact design enables deployment on resource-constrained edge devices (e.g., base stations or UAV relays). Since inference only involves forward propagation without iterative search, the model achieves predictable latency and low energy cost, further validating its suitability for real-time multicast routing in 6G environments. 
}

\section{Experiment}
\subsection{Experimental Setup}

To evaluate the performance of our proposed GNN-based multicast routing algorithm, we conduct experiments on synthetic network topologies generated using the NetworkX library. We consider two evaluation settings: (1) varying the number of nodes and users to assess the generalization ability of the model, and all graphs are generated with a fixed degree of 4 except the comparison for varying fixed node degre and varying average degree, ensuring moderate connectivity while preserving topological diversity., and (2) fixed-scale graphs with 50 nodes and 12 users, where the average node degree is controlled at 3, 4, 5, 6, and 7 to analyze performance under different topological densities. User demand levels are categorized as high ($1.0$), medium ($0.5$), and low ($0.25$), with approximately one-third of users randomly assigned to each level in every experiment. For each configuration, we generate 100 random graphs to account for statistical variance and ensure robustness of the evaluation. This simulates a realistic scenario of differentiated video service requirements in 6G networks.

We compare our method with the following baselines:
\begin{itemize}
    \item \textbf{Shortest Path Routing (Dijkstra)},
    \item \textbf{Genetic Algorithm (GA)},
    \item \textbf{Bee Colony Optimization (BCO)}\cite{Zheng2012An},
    \item \textbf{Dynamic Programming (DP)}, serving as a computationally expensive yet theoretically optimal reference.
    \item \textbf{Graph Attention Network (GAT)} baseline with demand-awareness, sequential routing, and flow reuse\cite{veličković2018graphattentionnetworks}.
\end{itemize}

All models are implemented in Python using NetworkX and PyTorch Geometric (PyG), and all experiments are conducted on a workstation equipped with an Intel(R) Xeon(R) Silver 4214R CPU and an NVIDIA RTX 3090 GPU. To highlight inference efficiency, we report the \textbf{logarithm (base 10)} of the routing time for each method, as the execution times across different algorithms vary by several orders of magnitude. This logarithmic transformation provides a more interpretable and visually balanced comparison on performance charts. In addition, we present qualitative visualizations of routing paths generated by various algorithms—including our method, GAT, shortest path, and the dynamic programming-based theoretical optimum—to intuitively illustrate structural efficiency, path reuse, and differentiated service routing in complex topologies. A comprehensive summary of all experimental results under different configurations is provided in Table~\ref{tab:summarize}.

\added{In our experiments, we set the hidden dimension $H=128$ and the number of attention heads $K=4$, which provide a practical balance between accuracy and efficiency. Increasing $H$ or $K$ yields marginal performance gains at the expense of memory and computation, whereas smaller values reduce the footprint but degrade accuracy. We train with Adam (learning rate $5\times 10^{-4}$), gradient–norm clipping at $1.0$, and a \texttt{MultiStepLR} scheduler with milestones every $500$ steps and decay factor $\gamma=0.96$. Each run uses a batch size of $16$ for $20$ epochs with $2{,}500$ steps per epoch. Training and validation graphs contain $30$ nodes and $9$ users; edges are sampled from an Erd\H{o}s--R\'enyi model with $p{=}0.10$ (training) and $p{=}0.08$ (validation). A fully connected virtual hub node is included; virtual edges take an effective cost of $10$ during reward computation. User–demand weights are $\{1.0, 0.5, 0.25\}$ for the first, middle, and last third of users, respectively, with index normalization by $\texttt{MAX\_USER}{=}20$. In the pointer module, logits are scaled by $10$ before the softmax, and we add $10^{-15}$ for numerical stability. Validation uses a batch size of $8$ with the same graph and user sizes.}
\begin{table*}[htbp]
\centering
\renewcommand{\arraystretch}{1.5}  
\resizebox{\textwidth}{!}{
\begin{tabular}{c|c|ccc|ccc|ccc|ccc|ccc|ccc}
\hline
\multicolumn{2}{c|}{Graph Settings} & \multicolumn{3}{c|}{Theoretical Optimal} & \multicolumn{3}{c|}{BCO} & \multicolumn{3}{c|}{GA} & \multicolumn{3}{c|}{Dijkstra} & \multicolumn{3}{c|}{GAT} & \multicolumn{3}{c}{\textbf{GPN (ours)}} \\
\hline
Setting & Count & \makecell{Cost} & \makecell{Delay} & \makecell{Cost-Delay\\Score}
 & \makecell{Cost} & \makecell{Delay} & \makecell{Cost-Delay\\Score} & \makecell{Cost} & \makecell{Delay} & \makecell{Cost-Delay\\Score} & \makecell{Cost} & \makecell{Delay} & \makecell{Cost-Delay\\Score} & \makecell{Cost} & \makecell{Delay} & \makecell{Cost-Delay\\Score} & \makecell{Cost} & \makecell{Delay} & \makecell{Cost-Delay\\Score}
\\
\hline
\multirow{5}{2.9cm}{\centering\arraybackslash\shortstack{Varying Node Count\\(12 users, degree 4)}} & 30 & 6.781 & 1.812 & 15.374 & 7.168 & 1.977 & 16.314 & 6.834 & 1.593 & 15.261 & 8.437 & -3.098 & \color[HTML]{9A0000} \textbf{13.776} & 7.418 & -0.712 & 14.124 & 7.200 & -0.585 & \color[HTML]{00009B} \underline{13.816} \\

& 35 & 7.296 & 1.843 & 16.436 & 7.755 & 2.045 & 17.554 & 7.357 & 1.646 & 16.359 & 9.255 & -3.065 & 15.445 & 7.995 & -0.677 & \color[HTML]{00009B} \underline{15.313} & 7.808 & -0.554 & \color[HTML]{9A0000} \textbf{15.063} \\

& 40 & 7.680 & 1.903 & 17.263 & 8.063 & 2.083 & 18.208 & 7.744 & 1.688 & 17.175 & 9.729 & -3.022 & 16.436 & 8.501 & -0.657 & \color[HTML]{00009B} \underline{16.345} & 8.158 & -0.535 & \color[HTML]{9A0000} \textbf{15.782} \\

& 45 & 7.853 & 1.973 & 17.680 & 8.336 & 2.120 & 18.793 & 7.933 & 1.713 & 17.579 & 9.942 & -3.011 & 16.873 & 8.624 & -0.654 & \color[HTML]{00009B} \underline{16.595} & 8.389 & -0.519 & \color[HTML]{9A0000} \textbf{16.259} \\

& 50 & 8.207 & 2.029 & 18.444 & 8.817 & 0.942 & 18.575 & 8.521 & 0.544 & \color[HTML]{00009B} \underline{17.587} & 10.758 & -2.994 & 18.523 & 9.251 & -0.627 & 17.875 & 8.781 & -0.515 & \color[HTML]{9A0000} \textbf{17.048} \\

\hline
\multirow{4}{2.9cm}{\centering\arraybackslash\shortstack{Varying Degree\\(50 nodes, 12 users)}} & 3 & 10.160 & 2.302 & 22.623 & 10.652 & 1.624 & 22.928 & 10.286 & 1.098 & \color[HTML]{00009B} \underline{21.669} & 12.652 & -3.007 & 22.298 & 12.005 & -0.550 & 23.461 & 10.930 & -0.461 & \color[HTML]{9A0000} \textbf{21.399} \\

& 4 & 8.150 & 2.044 & 18.344 & 8.619 & 1.485 & 18.724 & 8.326 & 1.086 & 17.737 & 10.566 & -2.997 & 18.135 & 8.983 & -0.625 & \color[HTML]{00009B} \underline{17.340} & 8.635 & -0.520 & \color[HTML]{9A0000} \textbf{16.749} \\

& 5 & 7.307 & 2.020 & 16.634 & 8.006 & 0.881 & 16.893 & 7.735 & 0.566 & 16.036 & 9.773 & -2.981 & 16.564 & 8.112 & -0.666 & \color[HTML]{00009B} \underline{15.558} & 7.871 & -0.555 & \color[HTML]{9A0000} \textbf{15.188} \\

& 6 & 6.812 & 2.010 & 15.634 & 7.482 & 0.868 & 15.831 & 7.239 & 0.604 & 15.082 & 8.843 & -2.951 & 14.735 & 7.566 & -0.686 & \color[HTML]{00009B} \underline{14.446} & 7.261 & -0.580 & \color[HTML]{9A0000} \textbf{13.942} \\

\hline
\multirow{4}{2.9cm}{\centering\arraybackslash\shortstack{Varying \\ Average Degree\\(50 nodes, 12 users)}} & 3 & 9.835 & 2.170 & 21.840 & 10.272 & 2.161 & 22.705 & 9.857 & 1.699 & 21.412 & 11.916 & -2.995 & \color[HTML]{00009B} \underline{20.837} & 12.253 & -0.557 & 23.949 & 10.411 & -0.467 & \color[HTML]{9A0000} \textbf{20.355} \\

& 4 & 7.995 & 2.252 & 18.242 & 8.513 & 0.884 & 17.910 & 8.258 & 0.555 & 17.071 & 9.881 & -3.005 & \color[HTML]{00009B} \underline{16.757} & 8.821 & -0.635 & 17.006 & 8.508 & -0.533 & \color[HTML]{9A0000} \textbf{16.483} \\

& 5 & 7.178 & 2.235 & 16.591 & 7.747 & 0.869 & 16.362 & 7.458 & 0.594 & 15.510 & 8.931 & -2.975 & \color[HTML]{00009B} \underline{14.886} & 8.113 & -0.660 & 15.565 & 7.659 & -0.561 & \color[HTML]{9A0000} \textbf{14.758} \\

& 6 & 6.830 & 2.247 & 15.907 & 7.413 & 0.884 & 15.709 & 7.191 & 0.639 & 15.021 & 8.644 & -2.943 & \color[HTML]{00009B} \underline{14.345} & 7.752 & -0.673 & 14.832 & 7.350 & -0.570 & \color[HTML]{9A0000} \textbf{14.130} \\

\hline
\multirow{9}{2.9cm}{\centering\arraybackslash\shortstack{Varying\\Number of User\\(50 nodes, degree 4)}} & 1 & 0.528 & -3.008 & \color[HTML]{00009B} \underline{-1.952} & 0.528 & -0.446 & 0.611 & 0.528 & -0.201 & 0.855 & 0.528 & -3.859 & \color[HTML]{9A0000} \textbf{-2.803} & 0.553 & -1.436 & -0.329 & 0.530 & -1.317 & -0.257 \\

& 2 & 1.437 & -2.252 & \color[HTML]{00009B} \underline{0.622} & 1.462 & -0.074 & 2.849 & 1.442 & -0.048 & 2.837 & 1.537 & -3.644 & \color[HTML]{9A0000} \textbf{-0.569} & 1.521 & -1.187 & 1.855 & 1.471 & -1.079 & 1.863 \\

& 3 & 3.169 & -1.839 & \color[HTML]{00009B} \underline{4.499} & 3.207 & 0.141 & 6.555 & 3.181 & 0.064 & 6.427 & 3.428 & -3.508 & \color[HTML]{9A0000} \textbf{3.348} & 3.358 & -1.053 & 5.663 & 3.220 & -0.901 & 5.489 \\

& 4 & 3.414 & -1.439 & \color[HTML]{00009B} \underline{5.389} & 3.474 & 0.310 & 7.257 & 3.432 & 0.152 & 7.016 & 3.880 & -3.408 & \color[HTML]{9A0000} \textbf{4.351} & 3.552 & -0.964 & 6.141 & 3.484 & -0.854 & 6.113 \\

& 5 & 3.924 & -0.940 & \color[HTML]{00009B} \underline{6.908} & 4.046 & 0.427 & 8.519 & 3.969 & 0.219 & 8.156 & 4.509 & -3.334 & \color[HTML]{9A0000} \textbf{5.683} & 4.105 & -0.899 & 7.311 & 4.052 & -0.791 & 7.312 \\

& 6 & 5.196 & -0.829 & \color[HTML]{00009B} \underline{9.562} & 5.452 & 0.531 & 11.435 & 5.302 & 0.284 & 10.889 & 6.195 & -3.272 & \color[HTML]{9A0000} \textbf{9.117} & 5.588 & -0.837 & 10.340 & 5.414 & -0.732 & 10.097 \\

& 9 & 6.810 & 0.587 & 14.208 & 7.222 & 0.773 & 15.218 & 6.991 & 0.441 & 14.423 & 8.683 & -3.101 & 14.265 & 7.448 & -0.711 & \color[HTML]{00009B} \underline{14.185} & 7.235 & -0.612 & \color[HTML]{9A0000} \textbf{13.859} \\

& 12 & 8.207 & 2.029 & 18.444 & 8.817 & 0.942 & 18.575 & 8.521 & 0.544 & \color[HTML]{00009B} \underline{17.587} & 10.758 & -2.994 & 18.523 & 9.251 & -0.627 & 17.875 & 8.781 & -0.515 & \color[HTML]{9A0000} \textbf{17.048} \\

& 15 & 9.600 & 3.740 & 22.940 & 10.666 & 1.089 & 22.421 & 10.089 & 0.649 & 20.828 & 12.540 & -2.889 & 22.192 & 10.665 & -0.561 & \color[HTML]{00009B} \underline{20.769} & 10.323 & -0.456 & \color[HTML]{9A0000} \textbf{20.189} \\

\hline
\multirow{3}{2.9cm}{\centering\arraybackslash\shortstack{Varying Extra\\User Size\\(50 nodes, 9 users,\\ degree 4)}} & 1 & 6.677 & 0.926 & 14.279 & 6.874 & 1.961 & 15.709 & 6.697 & 1.680 & 15.075 & 7.991 & -3.074 & 12.908 & 7.075 & -1.935 & \color[HTML]{9A0000} \textbf{12.214} & 7.072 & -1.692 & \color[HTML]{00009B} \underline{12.453} \\

& 2 & 7.195 & 1.473 & 15.862 & 7.425 & 2.017 & 16.867 & 7.255 & 1.712 & 16.222 & 8.598 & -3.043 & 14.154 & 7.516 & -1.449 & \color[HTML]{9A0000} \textbf{13.583} & 7.515 & -1.229 & \color[HTML]{00009B} \underline{13.801} \\

& 3 & 7.807 & 1.965 & 17.580 & 8.045 & 2.051 & 18.140 & 7.866 & 1.738 & 17.470 & 9.514 & -3.008 & 16.021 & 8.289 & -1.277 & \color[HTML]{9A0000} \textbf{15.300} & 8.220 & -1.102 & \color[HTML]{00009B} \underline{15.337} \\

\hline
\end{tabular}
}

\caption{Performance comparison under varying graph configurations. A composite Cost-Delay Score is used to jointly evaluate routing efficiency and latency, defined as $2 \times \text{Cost} + \log_{10}(\text{Delay})$.}
\label{tab:summarize}
\end{table*}

\subsection{\added{Illustrative Example}}
\added{
Consider a source $s$ and three destinations $u_1,u_2,u_3$ with heterogeneous demands 
$\bm{x}=[x_1,x_2,x_3]=[4k,2k,1k]$ (arbitrary units). 
By Theorem~\ref{theorem-1}, the optimal solution forms a tree; by Lemma~\ref{lemma-1}, the flow on any shared edge equals the maximum downstream demand. 
As illustrated in Fig.~\ref{fig:example}, $u_1$ and $u_2$ share an upstream link from $s$ with flow $\max\{x_1,x_2\}=4$, while $u_3$ is routed along a separate branch with flow $x_3=1$. 
The total transmission cost is
\[
C \;=\; e_{(s,a)}\!\cdot\!4 + e_{(a,u_1)}\!\cdot\!4 + e_{(a,u_2)}\!\cdot\!2 + e_{(s,b)}\!\cdot\!1 + e_{(b,u_3)}\!\cdot\!1 \;=\; 21(k),
\]
where the unit edge costs $e_{(\cdot,\cdot)}$ are annotated in Fig.~\ref{fig:example}.
}
\begin{figure}
    \centering
    \includegraphics[width=0.7\linewidth]{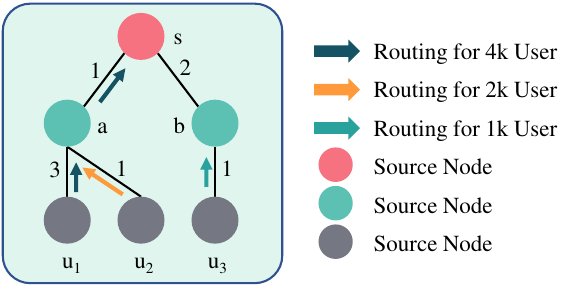}
    \caption{Routing Example}
    \label{fig:example}
\end{figure}
\added{
Our formulation naturally fits production deployments spanning \emph{CDN/edge multicast for live and on-demand video}, \emph{5G/6G multicast-broadcast services (MBS) and campus IPTV}, \emph{XR/holographic telepresence}, \emph{V2X HD map and perception updates}, and \emph{public-safety alerts} across administrative domains, where metro/core links are predominantly distance- and tariff-driven and thus well modeled by fixed weights; in all cases, path reuse across heterogeneous QoS aligns with our tree-optimality analysis and the sequential, demand-aware routing policy, enabling lower total flow under practical cost/latency constraints.
}

\subsection{Routing Cost Comparison across Algorithms}

To assess the routing efficiency in terms of total transmission cost, we compare our proposed method with four baseline algorithms under various network configurations. The evaluation focuses on how the total flow cost changes with respect to network scale, connectivity, and user group size.

\subsubsection*{1) Varying Node Count (Users = 12)}
We first analyze the routing cost as the number of nodes increases from 30 to 50 in steps of 5, while keeping the number of users fixed at 12. As shown in Fig.~\ref{fig:cost_nodecount}, our method consistently outperforms the shortest path and GAT-based approaches across all node sizes. Compared to the bee colony optimization baseline, GPN achieves similar performance, while it shows slightly higher cost than the genetic algorithm and the dynamic programming solution. Nonetheless, our model offers significantly better inference efficiency than these more computationally intensive methods, making it a practical alternative for large-scale scenarios. Notably, the shortest path method, despite incorporating a node reuse mechanism, still results in the highest transmission cost throughout. This is because it treats each user independently, failing to consider shared routes or the correlation between destination nodes, thus leading to redundant and fragmented paths.

\begin{figure}[!t]
    \centering
    \includegraphics[width=0.90\linewidth]{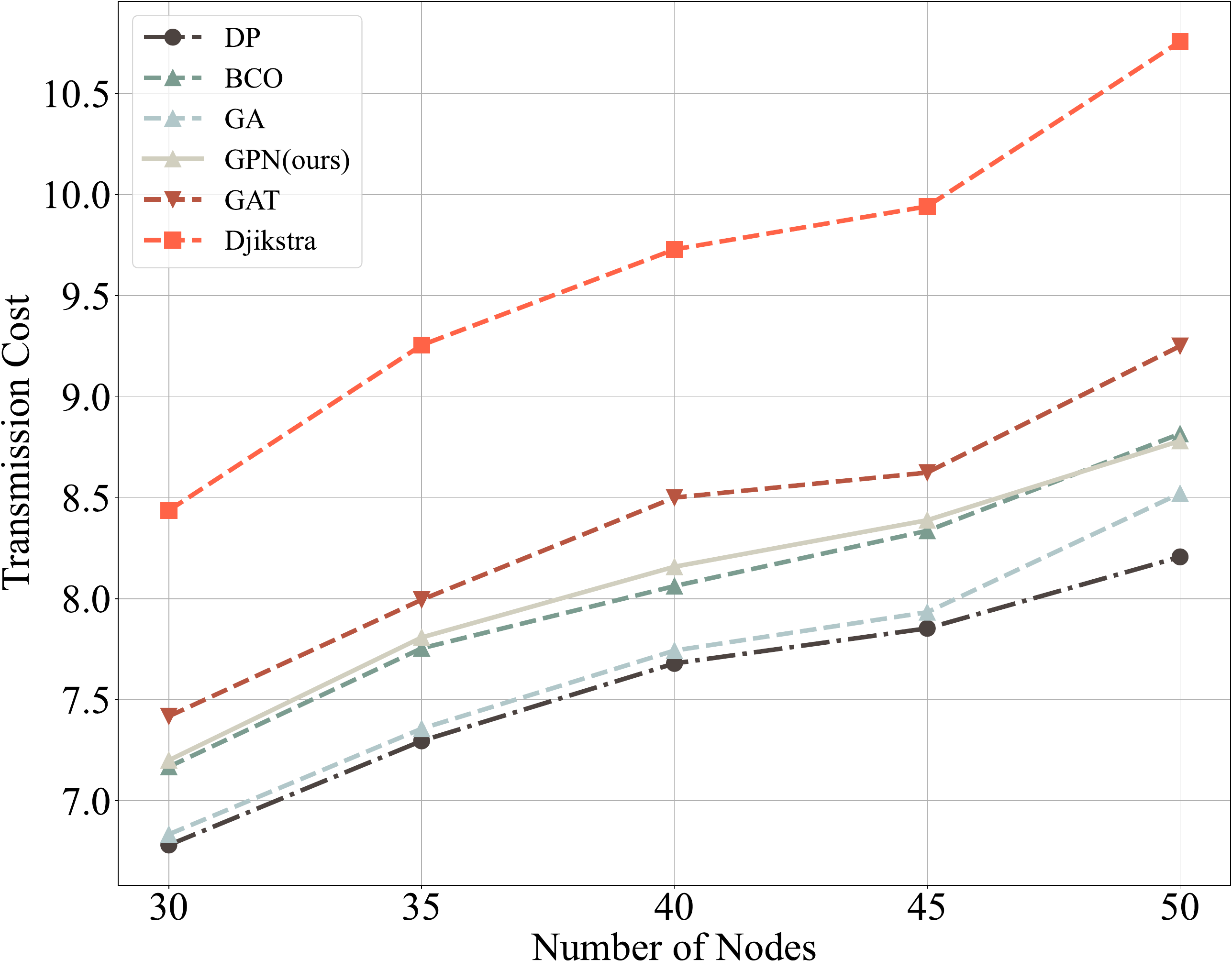}
    \caption{Total routing cost vs. number of nodes (users = 12).}
    \label{fig:cost_nodecount}
\end{figure}

        

\begin{figure}[!t]
    \centering
    \includegraphics[width=0.90\linewidth]{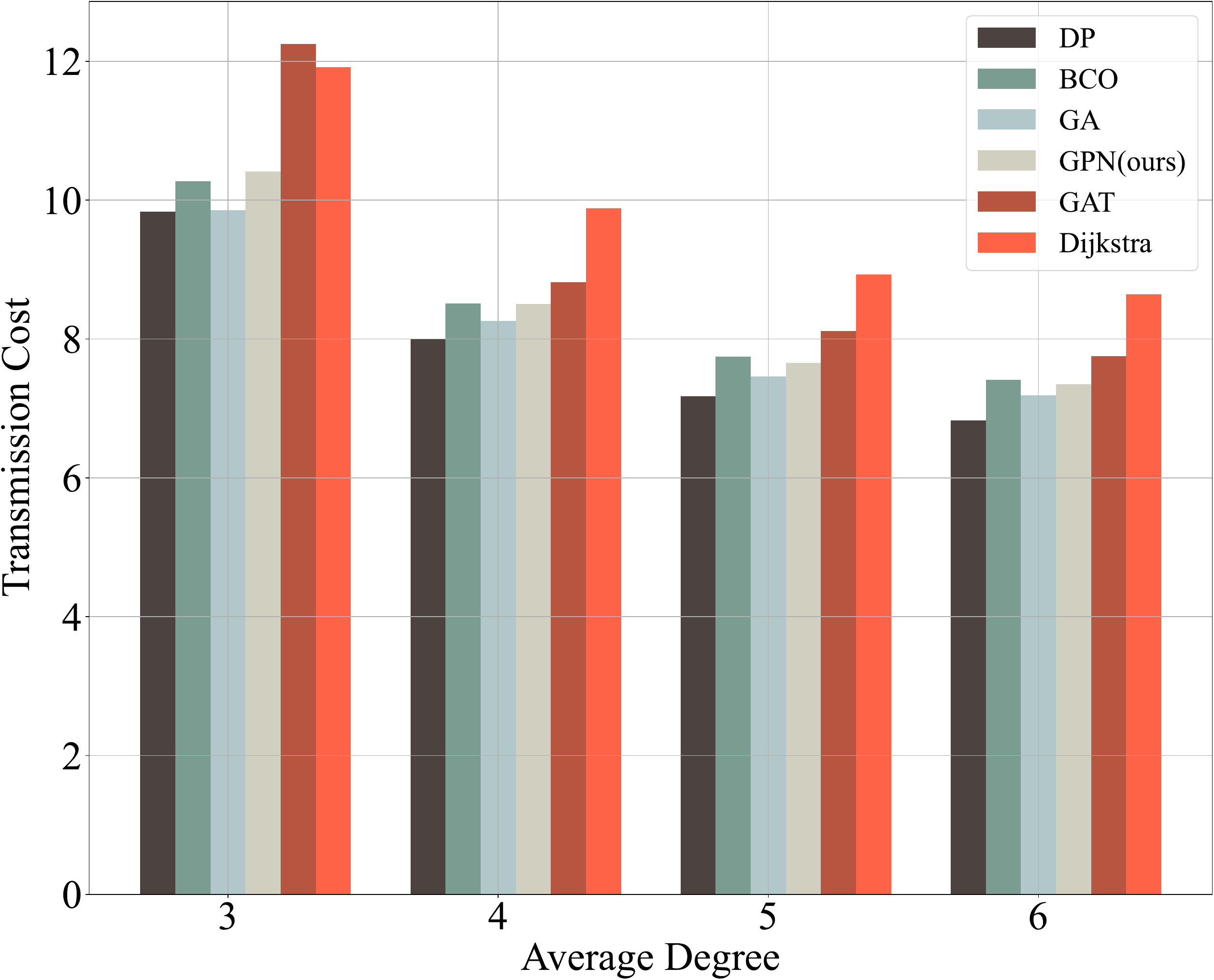}
    \caption{Total routing cost vs. average degree (nodes = 50, users = 12).}
    \label{fig:cost_avgdegree}
\end{figure}
\subsubsection*{2) Varying Average Degree (Nodes = 50, Users = 12)}
Next, we examine the effect of increasing average node degree from 4 to 6 under a fixed network size of 50 nodes and 12 users. The results are illustrated in Fig.~\ref{fig:cost_avgdegree}. As the connectivity becomes denser, all algorithms generally benefit from greater routing flexibility, resulting in reduced overall flow cost. Consistent with previous results, GPN outperforms shortest path and GAT by a significant margin and performs on par with BCO. It exhibits slightly higher cost than the GA and DP baselines, which remain the most cost-efficient overall. This again confirms that GPN achieves a strong balance between routing efficiency and computational feasibility. Notably, shortest path routing continues to yield the highest cost even with increased connectivity. Although we enhanced it with a basic node reuse mechanism, it remains fundamentally incapable of capturing correlations among users, leading to fragmented routing and poor overall efficiency.


\subsubsection*{3) Varying Fixed Node Degree (Nodes = 50, Users = 12)}

In addition to average connectivity, we evaluate scenarios with fixed node degrees of 3, 4, 5, and 6. As shown in Fig.~\ref{fig:cost_fixdeg}, our method consistently achieves competitive performance, with routing cost comparable to BCO and significantly lower than GAT and shortest path. It remains slightly behind GA and the DP reference, which continue to yield the lowest cost. These results confirm that our method maintains its effectiveness even under strict topological constraints. The shortest path approach again performs the worst despite the inclusion of a node reuse enhancement. Its inability to jointly optimize routes across multiple users results in inefficient path construction, especially in topologies with constrained degrees.

\begin{figure}[!t]
    \centering
    \includegraphics[width=0.90\linewidth]{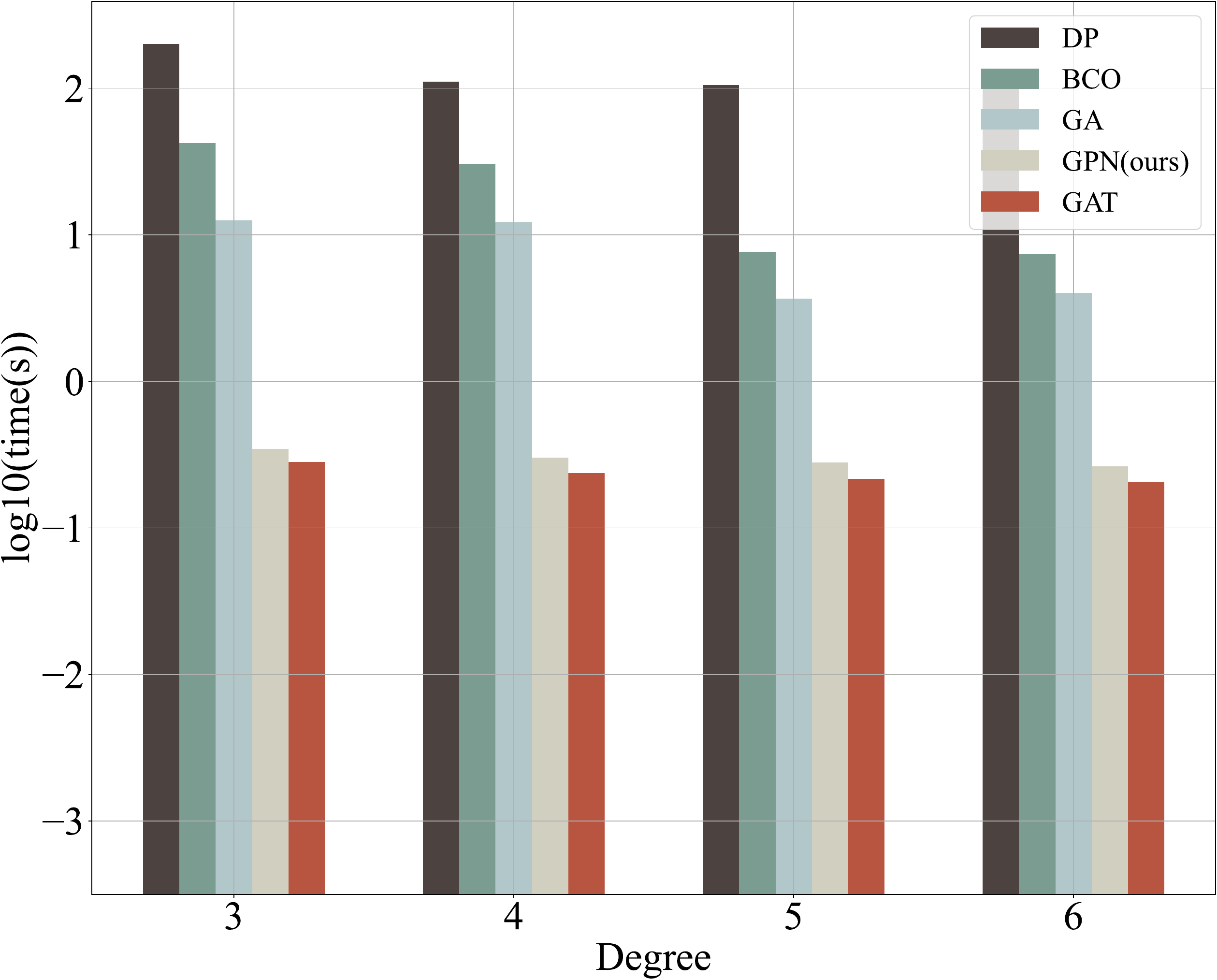}
    \caption{Total routing cost vs. fixed node degree (nodes = 50, users = 12).}
    \label{fig:cost_fixdeg}
\end{figure}


\subsubsection*{4) Small Number of Users (1 to 6)}
To evaluate the impact of user count on routing performance under light multicast demand, we test the system with a fixed network size of 50 nodes and vary the number of users from 1 to 6. As illustrated in Fig.~\ref{fig:scalability_few_users}, the total routing cost increases with the number of users for all methods. Notably, this increase is not strictly linear, which can be attributed to the heterogeneous service demands of the added users. Since each new user may request a different video quality, the extent to which their demand can be satisfied through flow reuse varies. This leads to fluctuations in the marginal increase of transmission cost, as observed in the figure. Our proposed method (GPN) consistently achieves performance close to the theoretical optimum (DP), with minimal deviation, and remains competitive with heuristic approaches such as GA and BCO. In contrast, GAT and shortest path routing incur substantially higher costs, underscoring their limited capability to coordinate multicast flows under differentiated low-demand conditions. These results demonstrate that even with a small number of users, GPN can effectively construct near-optimal multicast routes that efficiently leverage shared paths.

\begin{figure}[!t]
    \centering
    \includegraphics[width=0.90\linewidth]{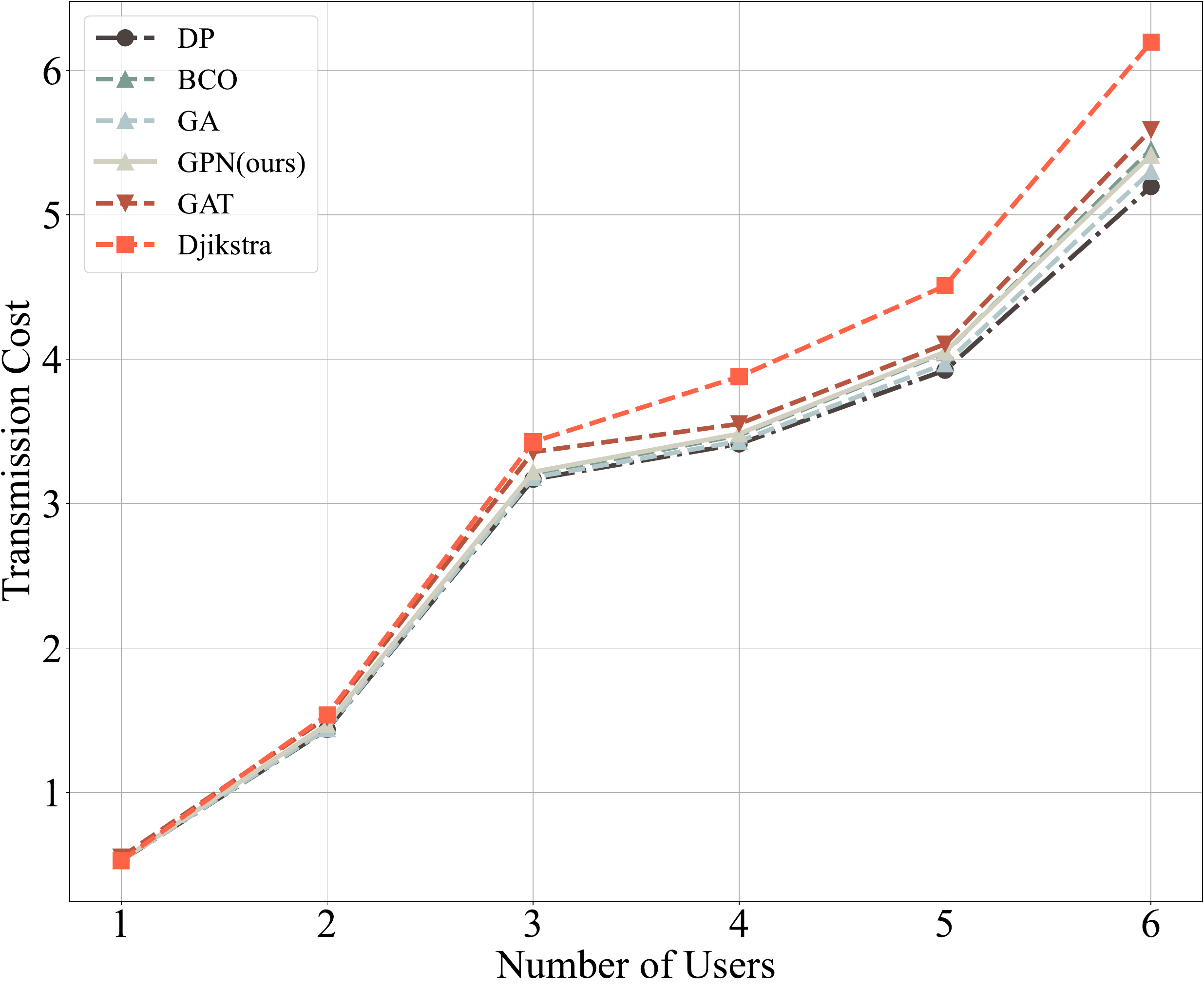}
    \caption{Routing cost vs. number of users (1–6), nodes = 50.}
    \label{fig:scalability_few_users}
\end{figure}


\subsubsection*{5) Large Number of Users (6 to 15, step = 3)}

We further extend the evaluation to scenarios with heavier multicast demands. The number of users is varied from 6 to 15 in steps of 3, under the same 50-node network. The results, shown in Fig.~\ref{fig:scalability_many_users}, demonstrate that GPN continues to scale effectively. While routing cost naturally increases with more destinations, the relative performance trends among algorithms remain consistent with previous settings. Specifically, GPN achieves routing cost close to BCO, and slightly higher than GA and the optimal DP baseline. In contrast, GAT and shortest path methods exhibit significantly steeper cost increases, reflecting their limited ability to aggregate delivery paths. These findings confirm that GPN provides a strong balance between efficiency and performance, making it well-suited for dynamic and large-scale multicast scenarios.

\begin{figure}[!t]
    \centering
    \includegraphics[width=0.90\linewidth]{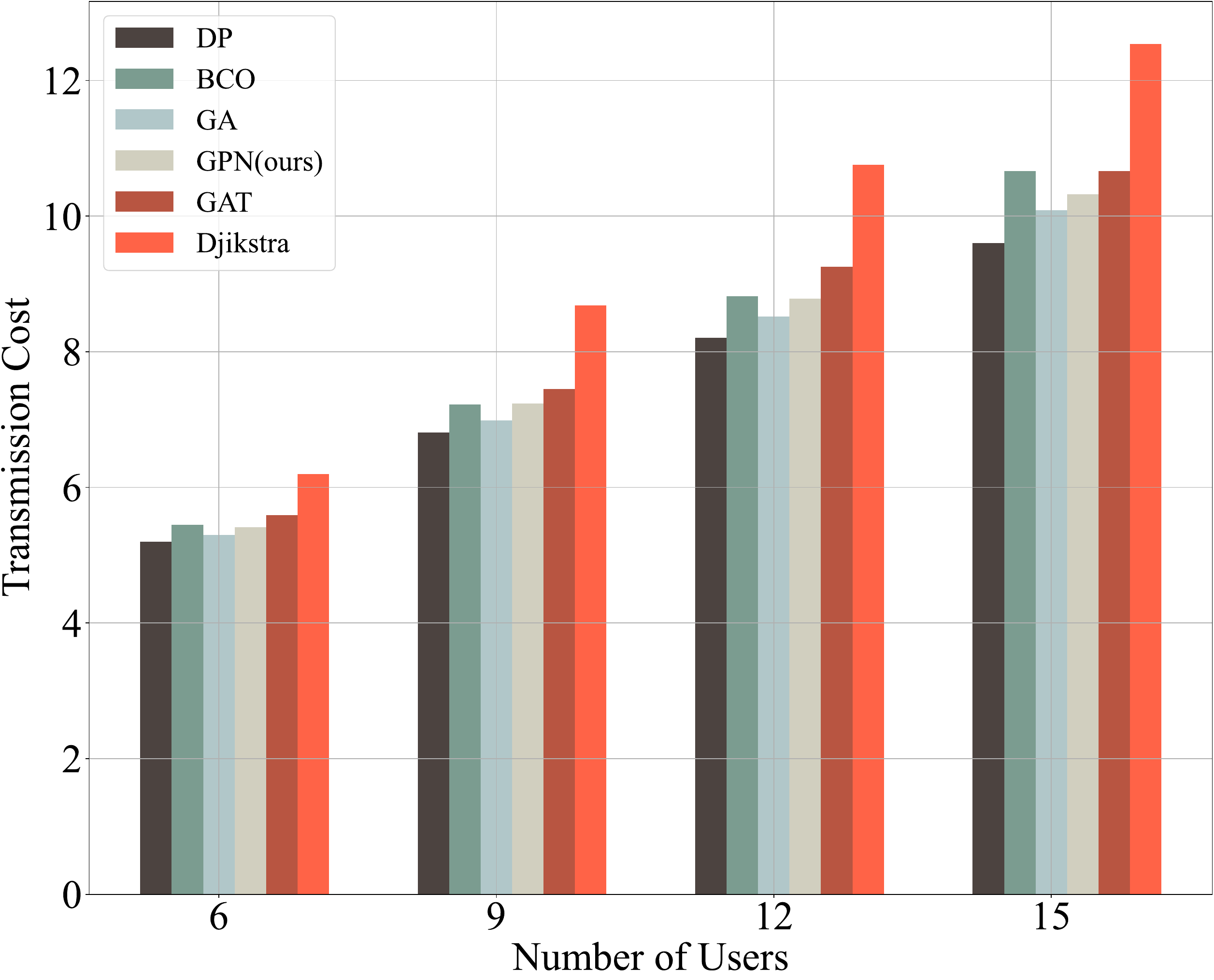}
    \caption{Routing cost vs. number of users (6–15), nodes = 50.}
    \label{fig:scalability_many_users}
\end{figure}


\subsection{Statistical Comparison under Fixed Conditions}

To provide a deeper statistical understanding of each algorithm’s performance under consistent network settings, we visualize the distribution of total routing costs using violin plots. These plots illustrate not only the mean and variance, but also the density of results over 20 random graph instances, revealing the stability and reliability of each method.
\begin{figure*}[!t]
    \centering
    \includegraphics[width=0.9\linewidth]{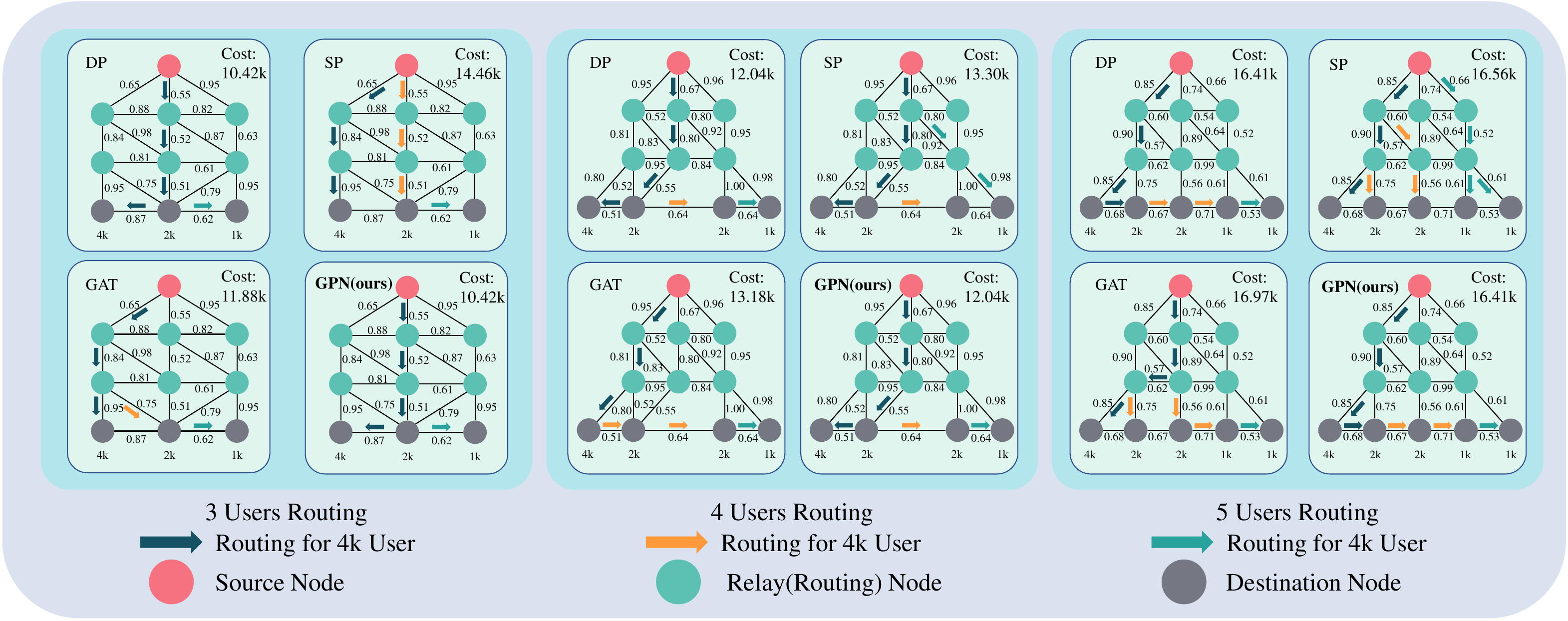}
    \caption{Routing path comparison under a fixed 50-node topology with 3, 4, and 5 users. Each column corresponds to a user configuration; each row shows the result of a different algorithm.}
    \label{fig:path_compare_smallusers}
\end{figure*}
\subsubsection*{1) Scenario: 30 Nodes, 12 Users}
As shown in Fig.~\ref{fig:violin_30n_12u}, GPN demonstrates a compact and low-cost distribution, with mean performance comparable to GA and BCO, both of which also show strong results. In contrast, the shortest path method produces much higher average cost and greater variability, as it independently constructs routes for each user without considering the overall multicast structure. The dynamic programming solution achieves the best performance overall, with the most concentrated distribution and the lowest mean cost, serving as a theoretical optimum for reference. GPN’s results closely follow those of DP, confirming its near-optimal behavior even in smaller topologies.

\begin{figure}[!t]
    \centering
    \includegraphics[width=0.90\linewidth]{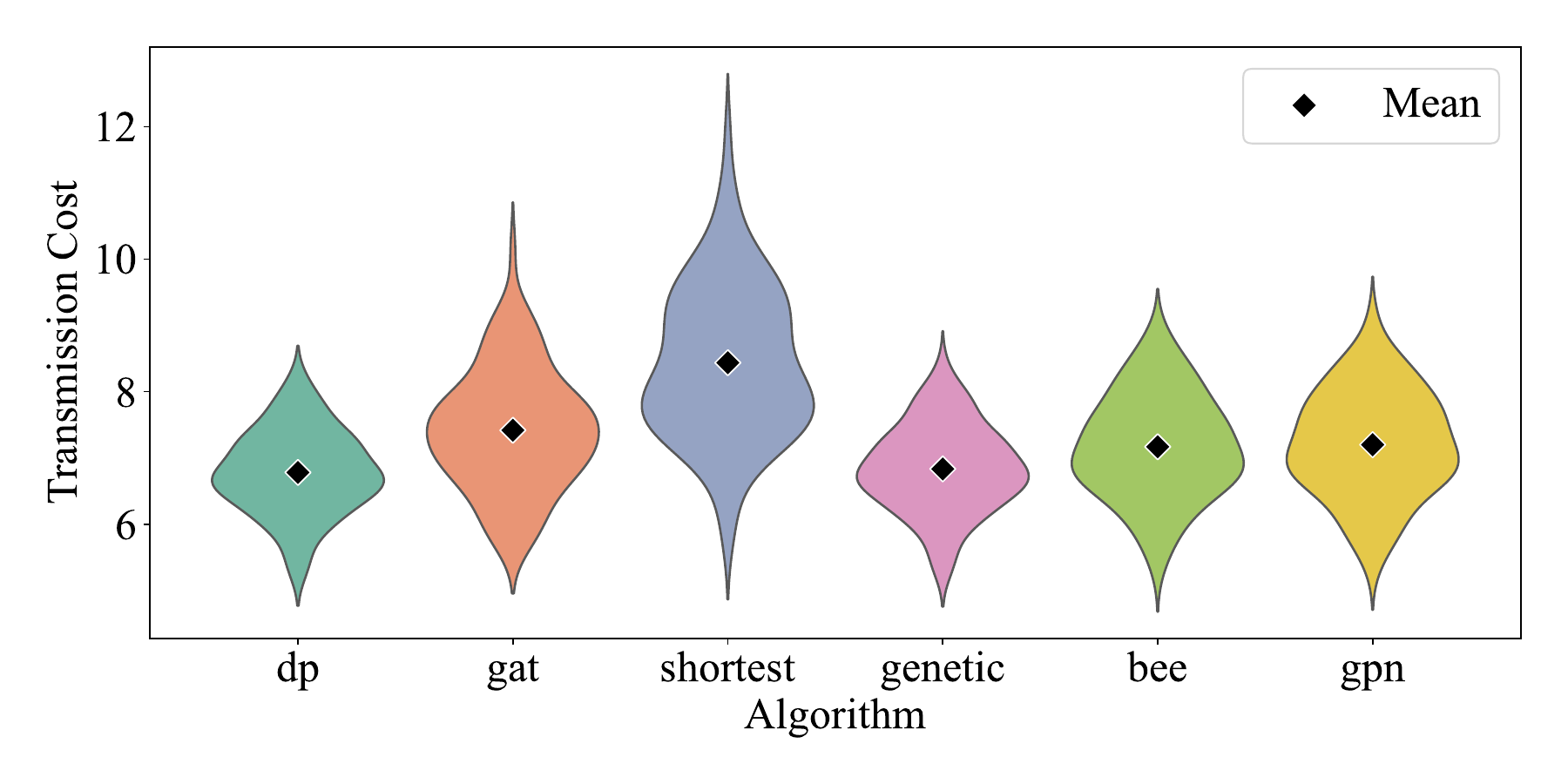}
    \caption{Routing cost distribution under 30 nodes and 12 users.}
    \label{fig:violin_30n_12u}
\end{figure}

\subsubsection*{2) Scenario: 50 Nodes, 12 Users}
In Fig.~\ref{fig:violin_50n_12u}, GPN maintains a tight and low-cost distribution in larger network settings. It performs comparably to GA and BCO, both of which continue to exhibit good routing efficiency and stable outcomes. Compared to GAT, GPN shows a more concentrated distribution and slightly lower average cost, highlighting its robustness and consistent behavior. Again, the DP baseline achieves the lowest cost with the smallest variance, and GPN remains very close in both average performance and spread.

\begin{figure}[!t]
    \centering
    \includegraphics[width=0.90\linewidth]{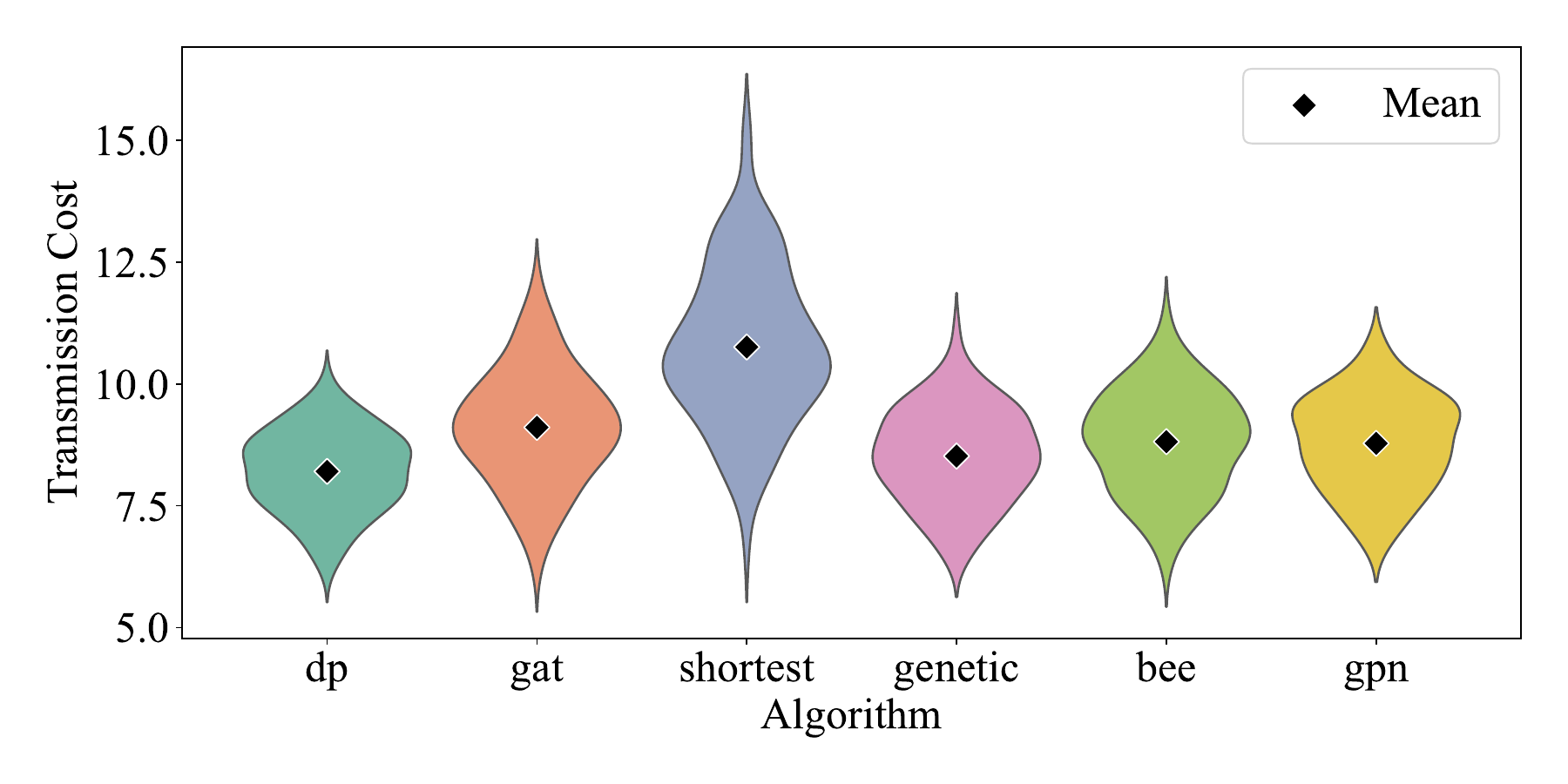}
    \caption{Routing cost distribution under 50 nodes and 12 users.}
    \label{fig:violin_50n_12u}
\end{figure}


Under higher user demand, Fig.~\ref{fig:violin_50n_15u} shows that GPN maintains strong performance. Although the routing cost increases for all methods, GPN still achieves low mean cost and compact variance, comparable to GA and BCO. The shortest path method continues to perform the worst, showing both higher mean cost and broader spread. The DP baseline remains the most efficient and stable, and GPN’s distribution again aligns closely with it, demonstrating the scalability and consistency of our approach even in complex scenarios.

\begin{figure}[!t]
    \centering
    \includegraphics[width=0.90\linewidth]{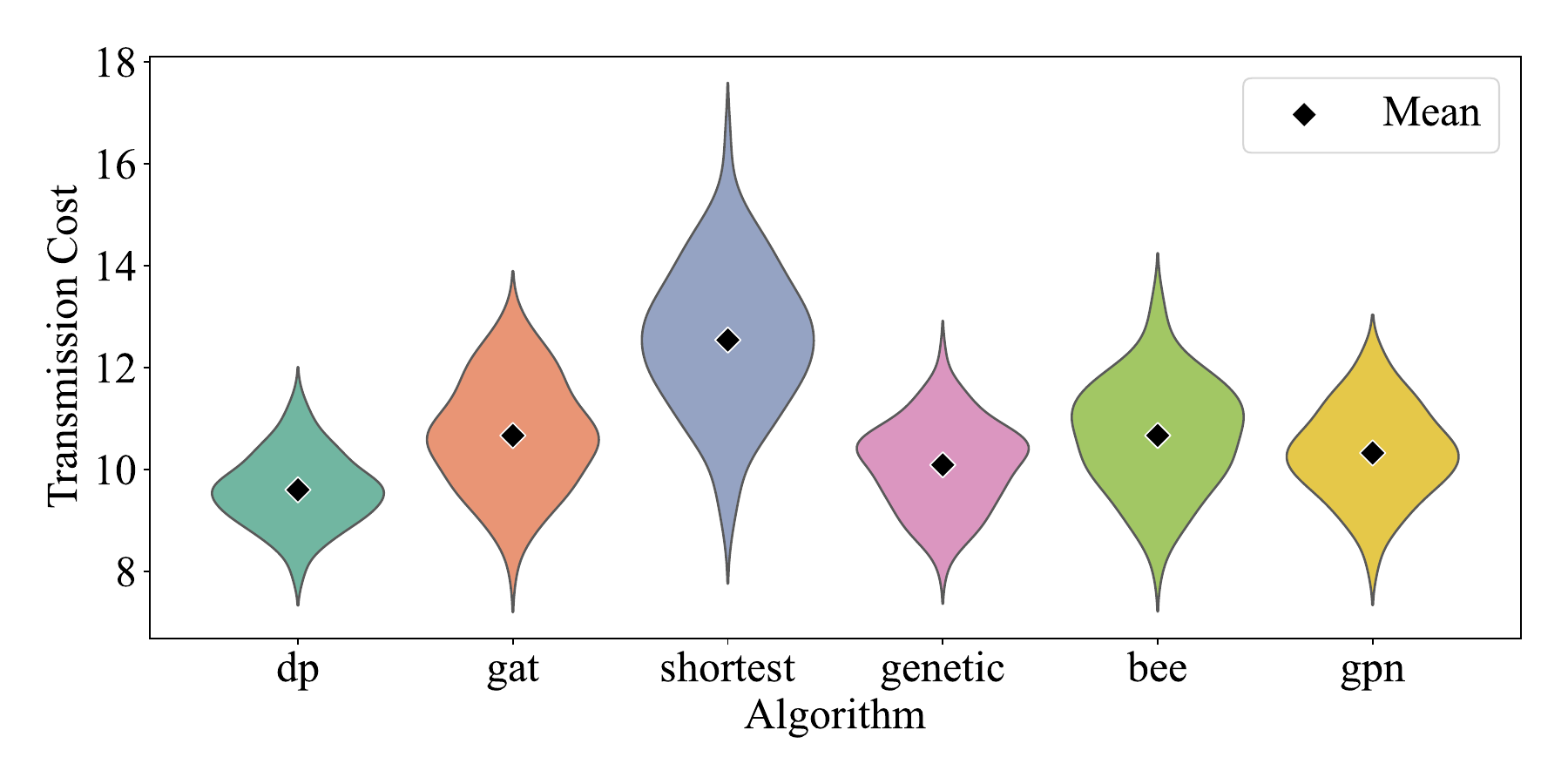}
    \caption{Routing cost distribution under 50 nodes and 15 users.}
    \label{fig:violin_50n_15u}
\end{figure}

To provide an intuitive comparison of the routing structures produced by different algorithms, we present a visual analysis based on several fixed small-scale networks. In this controlled setting, we manually select three representative multicast configurations with 3, 4, and 5 destination users, respectively. The underlying graph topology and user placements are carefully chosen to expose the strengths and limitations of each algorithm’s routing behavior in a visually interpretable way. For illustration purposes, the routing costs are generated with randomness to focus on comparing the overall structure, aggregation quality, and routing strategy exhibited by each method. As shown in Fig.~\ref{fig:path_compare_smallusers}, the shortest path algorithm consistently yields the highest routing cost across all configurations. Its disjoint and greedy unicast routing strategy lacks any consideration of shared structure or user correlation, resulting in extensive path redundancy. While the graph attention network demonstrates some capability for reusing intermediate nodes, it often fails to discover globally optimal or compact routes due to limited ability to fully capture graph-wide dependencies. Its routing structures appear fragmented and inefficient, especially as the number of users increases.

In contrast, our proposed GPN model exhibits robust and coherent routing patterns that are both well-aggregated and cost-effective. Across all three multicast settings, GPN closely follows the routing structures produced by dynamic programming, which serves as the theoretical optimum. GPN is able to implicitly model user relationships and topological regularities, allowing it to construct efficient routes with minimal overhead, even without explicitly enumerating optimal paths as DP does. This visual comparison further validates the effectiveness of GPN in generating centralized, reusable multicast paths that scale well with user demand, offering a practical balance between structural quality and computational efficiency.

\subsection{Time Consumption Comparison}

Inference time is a critical metric for evaluating the practical deployment of multicast routing algorithms, especially in real-time systems. In this section, we compare the execution time of all methods under varying numbers of users, while keeping the network size fixed at 50 nodes. Since the time differences between algorithms span several orders of magnitude, we report the results using the \textbf{logarithm (base 10)} of execution time.

As shown in Fig.~\ref{fig:time_users_all}, the proposed GPN model achieves consistently low inference time across all user counts. From 1 to 6 users, GPN, GAT, and shortest path routing maintain sub-second execution times, with only minor increases. For 9, 12, and 15 users, GPN continues to scale effectively, while heuristic methods such as GA, BCO, and DP exhibit rapidly increasing runtime due to their high algorithmic complexity.

Among all methods, GPN offers the best trade-off between routing quality and computational cost. It is significantly faster than GA, BCO, and DP, and achieves even lower routing cost than GAT and shortest path methods. These results confirm the suitability of our method for real-time, large-scale multicast routing in next-generation networks.

\begin{figure}[!t]
    \centering
    \includegraphics[width=0.8\linewidth]{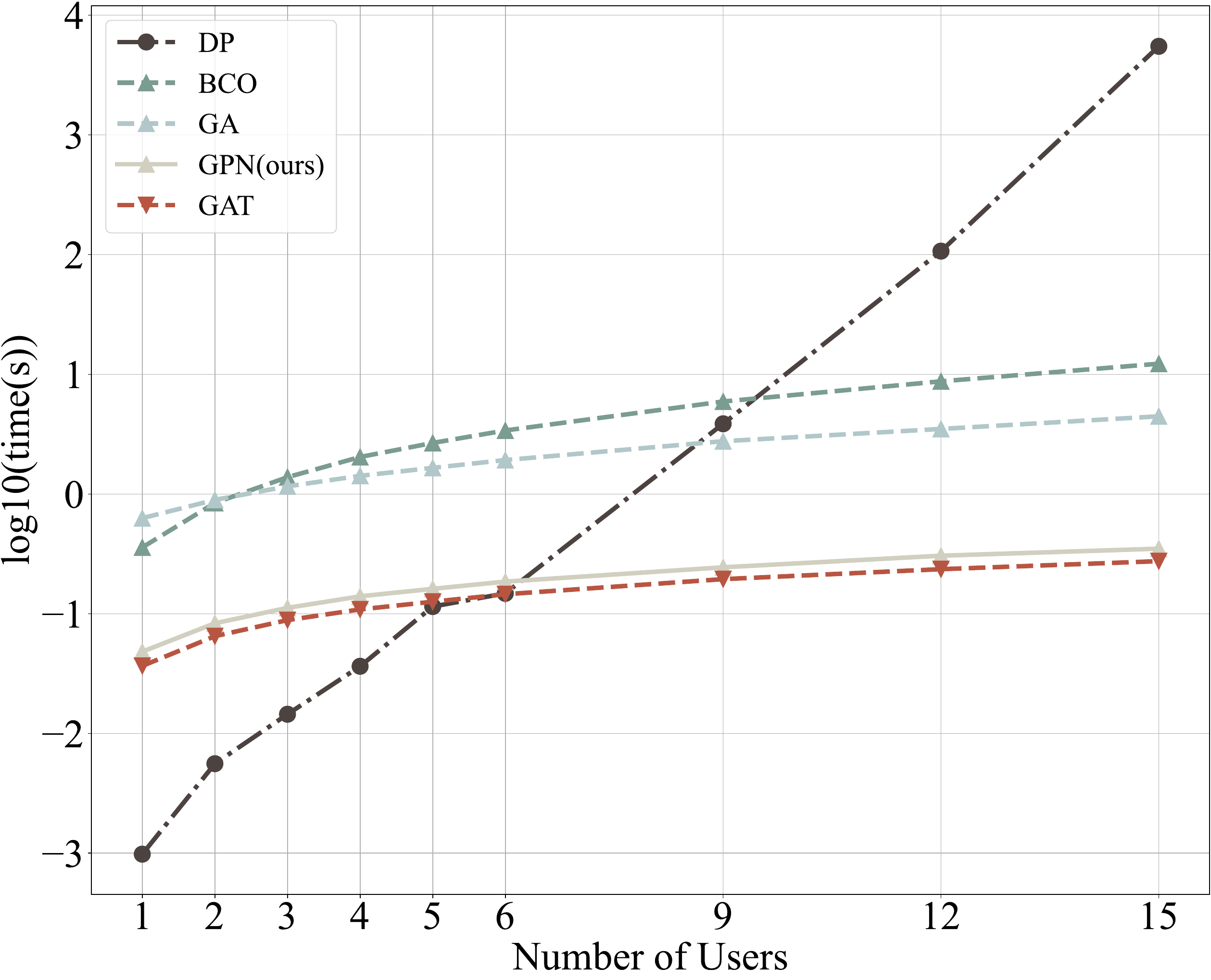}
    \caption{Log$_{10}$ execution time vs. number of users (1–6 and 9–15), nodes = 50.}
    \label{fig:time_users_all}
\end{figure}

\subsection{Incremental Routing with Dynamic User Addition}

\begin{figure}[!t]
    \centering
    \includegraphics[width=0.8\linewidth]{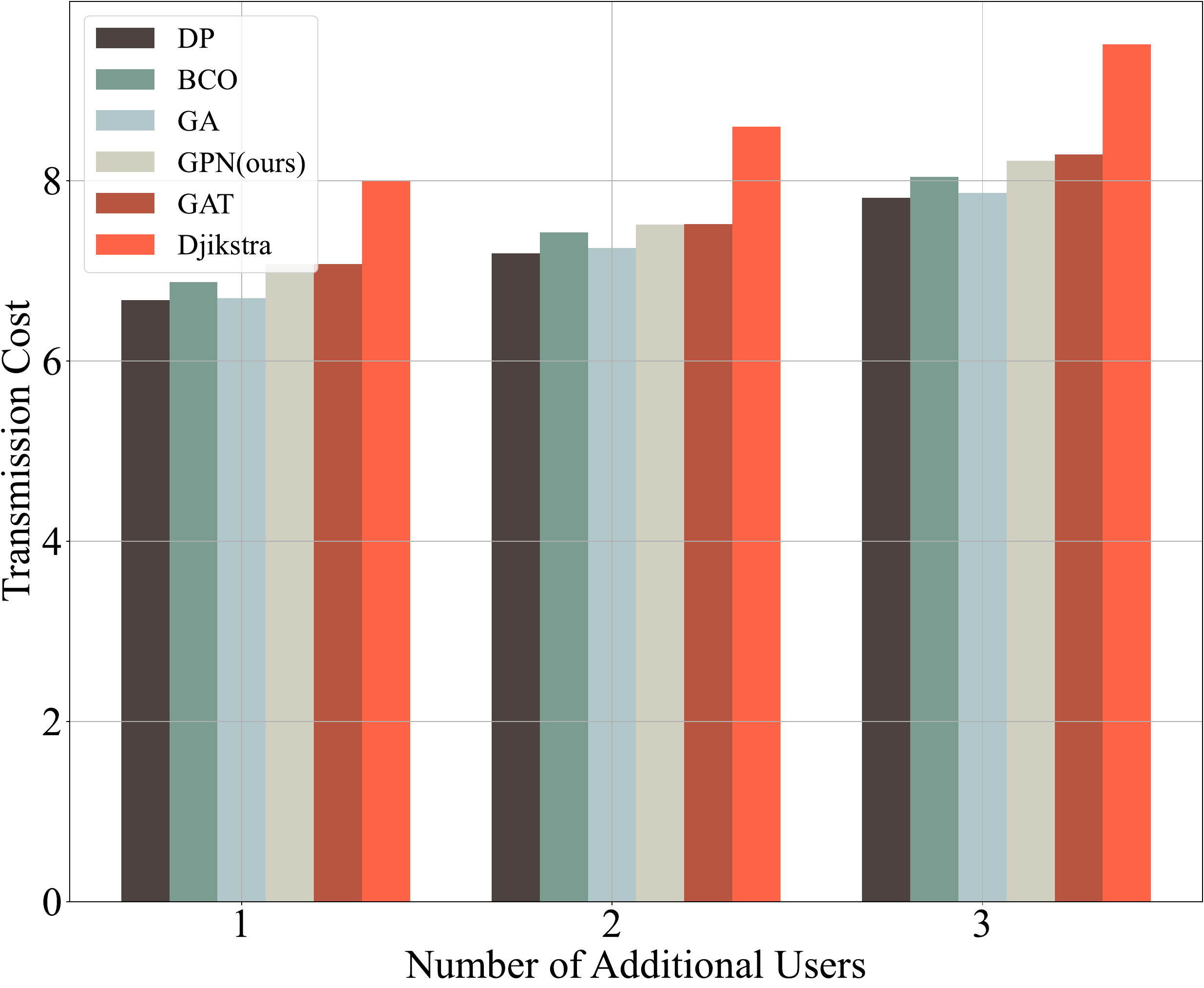}
    \caption{Routing cost under dynamic user addition (initial: 9 users, nodes = 50).}
    \label{fig:incremental_cost}
\end{figure}


\begin{figure}[!t]
    \centering
    \includegraphics[width=0.8\linewidth]{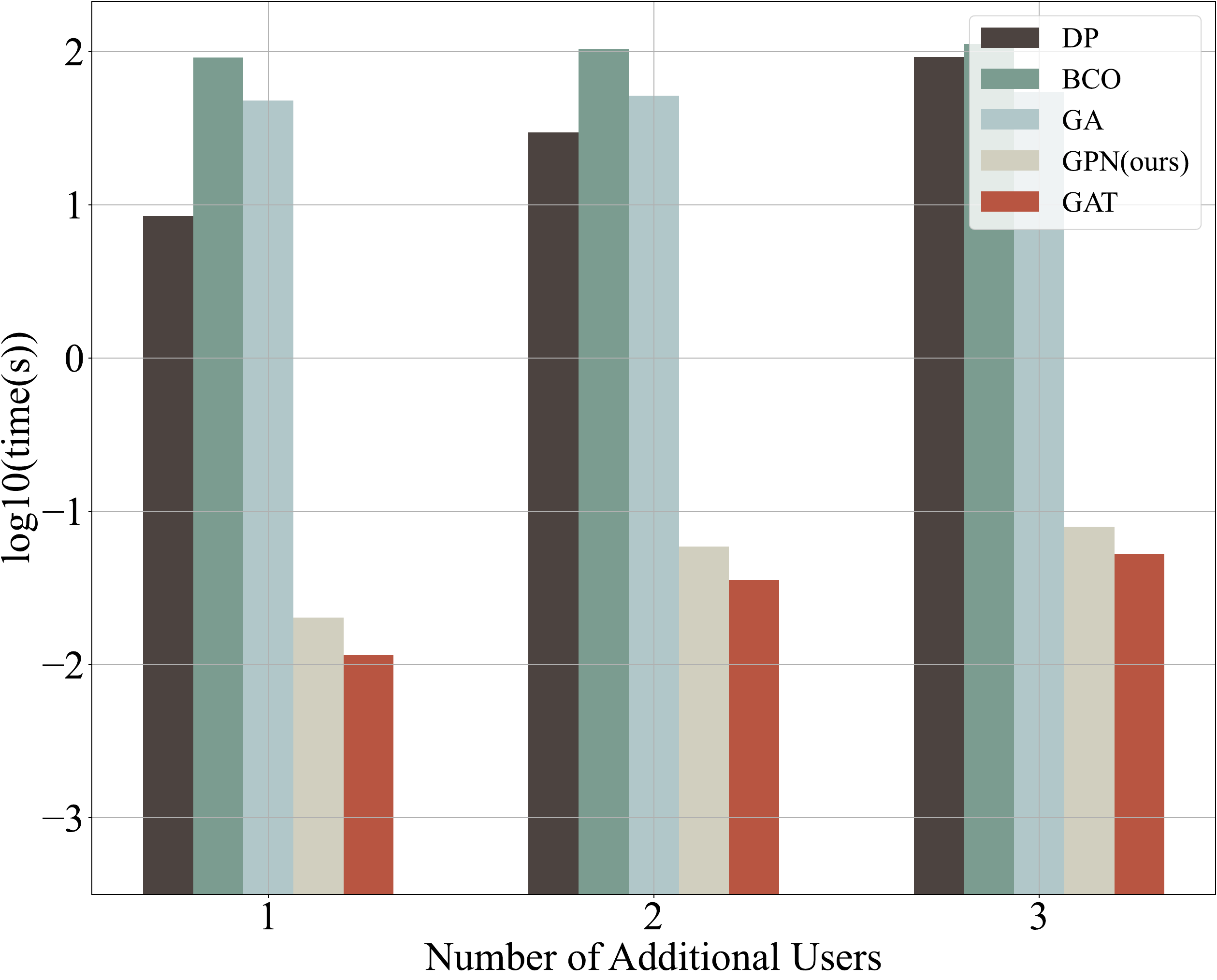}
    \caption{Log$_{10}$ execution time under dynamic user addition (initial: 9 users, nodes = 50).}
    \label{fig:incremental_time}
\end{figure}


In practical multicast scenarios, user demands may arrive sequentially or in bursts. Traditional algorithms, such as GA, BCO, and even shortest path routing, must recompute the entire routing structure from scratch whenever the destination set changes, resulting in significant latency and computational overhead. In contrast, our proposed GPN model supports efficient \textbf{incremental rerouting}: when new users are added, the model can rapidly integrate them into the existing delivery structure by intelligently reusing and merging previously optimized paths. This design enables low-latency updates while avoiding redundant computation. Although graph attention networks can also support similar sequential routing mechanisms, our experiments reveal that GAT incurs slightly higher routing cost in the incremental setting, indicating its limited adaptability to dynamic user updates.

To evaluate this property, we simulate dynamic user addition on a 50-node graph. Starting from an initial routing plan for 9 users, we incrementally add 1, 2, and 3 new users and compare the results across multiple algorithms. Fig.~\ref{fig:incremental_cost} and Fig.~\ref{fig:incremental_time} illustrate the resulting changes in routing cost and execution time, respectively, while Tables~\ref{tab:summarize} summarize the statistical results. These comparisons demonstrate that GPN not only achieves rapid adaptation to user-set changes but also maintains low transmission cost, outperforming both traditional heuristics and neural baselines such as GAT in terms of efficiency and scalability.

\subsection{\added{Ablation Study}}

\added{
We conduct an ablation experiment on a graph with 30 nodes and 12 users, where the full GPN model is compared with variants that replace the GAT encoder with a generic GNN, remove the LSTM decoder, or substitute the attention-based pointer with an MLP scorer. As shown in Fig.~\ref{fig:ablation}, replacing the GAT with a generic GNN causes a large increase in transmission loss, highlighting the critical role of multi-head attention in capturing graph structure. In contrast, removing the LSTM or replacing the attention-based pointer with an MLP only leads to minor performance drops, indicating that these modules have a smaller impact on overall effectiveness.
}
\begin{figure}[t]
    \centering
    \includegraphics[width=0.8\linewidth]{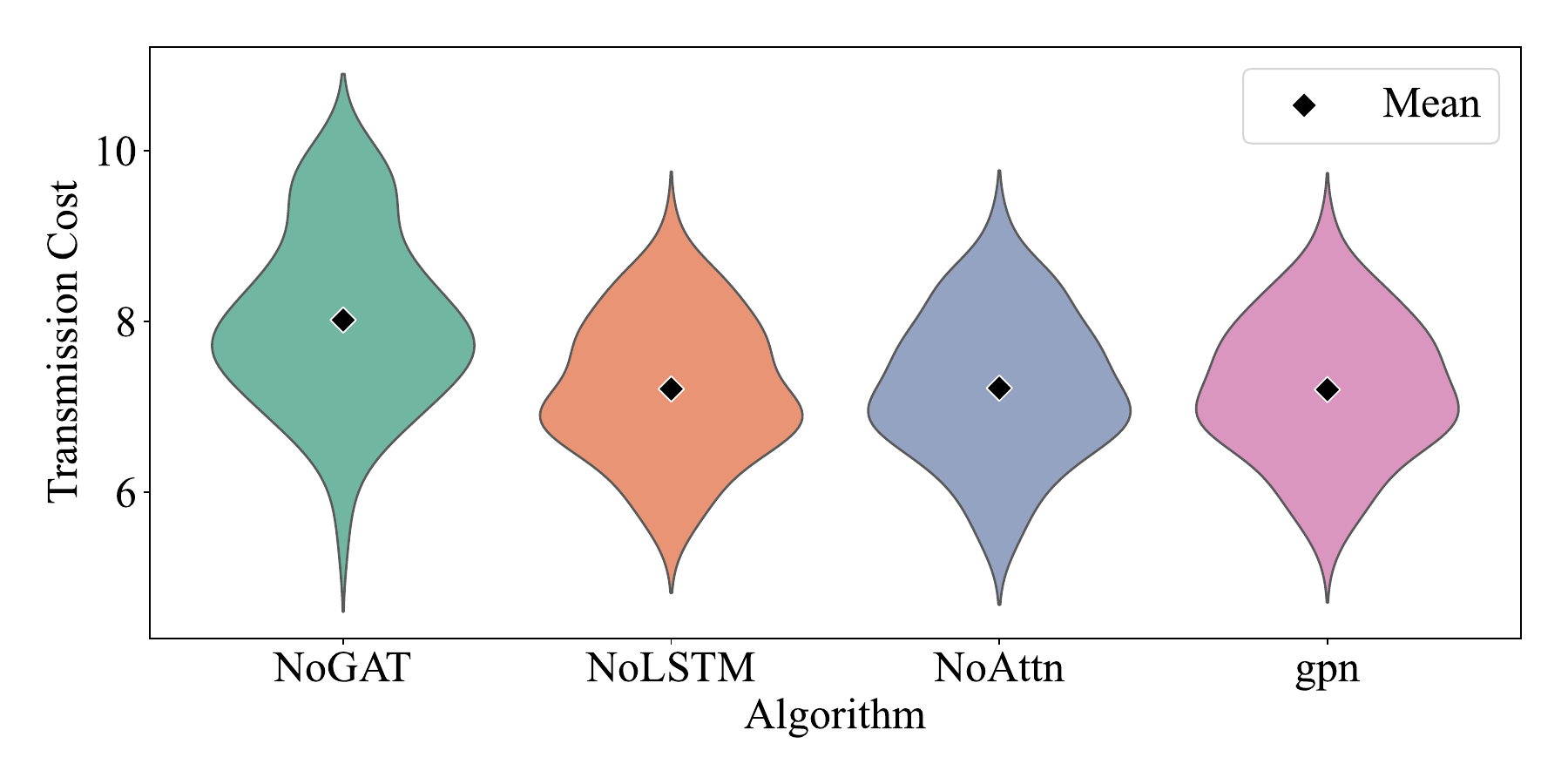}
    \caption{\added{Ablation study on the 30-node, 12-user case.}}
    \label{fig:ablation}
\end{figure}

\section{Conclusion}
In this paper, we have presented a robust and scalable multicast routing framework leveraging graph neural networks and reinforcement learning to address the critical challenges of heterogeneous QoS demands, network dynamism, and scalability limitations inherent in traditional multicast routing methods. Comprehensive experiments underscore the superior performance in minimizing transmission costs and enhancing computational efficiency for real-time deployment of our approach, which is well-suited for future multimedia streaming applications within the evolving landscape of 6G networks. In future work, we will explore how to extend our proposed method to address multi-source multicast and coordinated scheduling of multi-resolution video streams.

\bibliography{ref}
\bibliographystyle{IEEEtran}

\ifCLASSOPTIONcaptionsoff
  \newpage
\fi

\end{document}